\documentclass[12pt,onecolumn,technote]{IEEEtran}

\usepackage{vmargin}
\renewcommand{\baselinestretch}{1.0}
\usepackage{cite}
\usepackage{graphicx}
\DeclareGraphicsExtensions{.pdf,.jpeg,.png}
\usepackage[cmex10]{amsmath}
\usepackage{multirow}
\usepackage{balance}
\usepackage{enumerate}

\usepackage{amssymb}
\usepackage{algorithmicx, algorithm ,algpseudocode}
\usepackage{threeparttable}

\usepackage{dsfont}
\usepackage{booktabs}
\usepackage{bm}
\usepackage{balance}

\usepackage{array}
\usepackage{booktabs}
\usepackage{bm}
\usepackage{mdwmath}
\usepackage{mdwtab}

\usepackage{subfigure}

\usepackage{url}
\usepackage{lipsum}

\usepackage{amsthm}
\newtheorem{theorem}{\textbf{Theorem}}
\newtheorem{lemma}{\textbf{Lemma}}
\newtheorem{definition}{\textbf{Definition}}

\newtheorem{proposition}{\textbf{Proposition}}

\begin{document}

\title{Least Cost Influence Maximization Across Multiple Social Networks}

\author{Huiyuan~Zhang,
		Dung~T.~Nguyen,
        Soham~Das,
        Huiling~Zhang
        and~My~T.~Thai\\
\IEEEauthorblockA{Department of Computer and Information Science and Engineering
\\University of Florida, Gainesville, Florida 32611\\
Email: \{huiyuan, dtnguyen, sdas, huiling, mythai\}@cise.ufl.edu}
\thanks{This article has been published in IEEE/ACM Transactions on Networking, 24(2), 929-939, March 12, 2015, DOI: 10.1109/TNET.2015.2394793. Copyright (c) 2015 IEEE. Personal use of this material is permitted. Permission from IEEE must be obtained for all other uses, in any current or future media, including reprinting/republishing this material for advertising or promotional purposes, creating new collective works, for resale or redistribution to servers or lists, or reuse of any copyrighted component of this work in other works.
H. Zhang and D. T. Dung are co-first authors of this work.}}

\maketitle

\begin{abstract}
Recently in Online Social Networks (OSNs), the \textit{Least Cost Influence} (LCI) problem has become one of the central research topics. It aims at identifying a minimum number of seed users who can trigger a wide cascade of information propagation. Most of existing literature investigated the LCI problem only based on an individual network. However, nowadays users often join several OSNs such that information could be spread across different networks simultaneously. Therefore, in order to obtain the best set of seed users, it is crucial to consider the role of overlapping users under this circumstances.

In this article, we propose a unified framework to represent and analyze the influence diffusion in multiplex networks. More specifically, we tackle the LCI problem by mapping a set of networks into a single one via lossless and lossy coupling schemes. The lossless coupling scheme preserves all properties of original networks to achieve high quality solutions, while the lossy coupling scheme offers an attractive alternative when the running time and memory consumption are of primary concern. Various experiments conducted on both real and synthesized datasets have validated the effectiveness of the coupling schemes, which also provide some interesting insights into the process of influence propagation in multiplex networks.

\end{abstract}


\section{Introduction}
In the recent decade, the popularity of online social networks, such as Facebook, Google+, Myspace and Twitter etc., has created a new major communication medium and formed a promising landscape for information sharing and discovery. On average \cite{osnstatistics}, Facebook users spend 7 hours and 45 minutes per person per month on interacting with their friends ; 3.2 billion likes and comments are posted every day on Facebook; 340 million tweets are sent out everyday on Twitter. Such engagement of online users fertilizes the land for information propagation to a degree which has never been achieved before in the mass media. More importantly, OSNs also inherit one of the major properties of real social networks -- the word-of-mouth effect, in which personal opinion or decision can be reshaped or reformed through influence from friends and colleagues. Recently, motivated by the significant effect of viral marketing, OSNs have been the most attractive platforms to increase brand awareness of new products as well as strengthen the relationship between customers and companies. In general, the ultimate goal is to find the least advertising cost set of users which can trigger a massive influence.

Along with the fast development of all existing OSNs, there have been quite a number of users who maintain several accounts simultaneously, which allow them to propagate information across different networks. For example, Jack, a user of both Twitter and Facebook, knew a new book from Twitter. After reading it, he found it very interesting and shared this book with friends in Facebook as well as Twitter. This can be done by configuring both of the accounts to allow automatically posting across different social networks. As a consequence, the product information is exposed to his friends and further spreads out on both networks. If we only focus on an individual network, the spread of the information is estimated inaccurately. As shown in Fig. \ref{fig:infoPropagation}, the fraction of overlapping users is considerable. Therefore considering the influence only in one network fails to identify the most influential users, which motivates us to study the problem in multiplex networks where the influence of users is evaluated based on all OSNs in which they participate.

\begin{figure}
\centering
	\subfigure[Auto post from Facebook to Twitter] { 
		\centering
		\includegraphics[width=0.3\columnwidth]{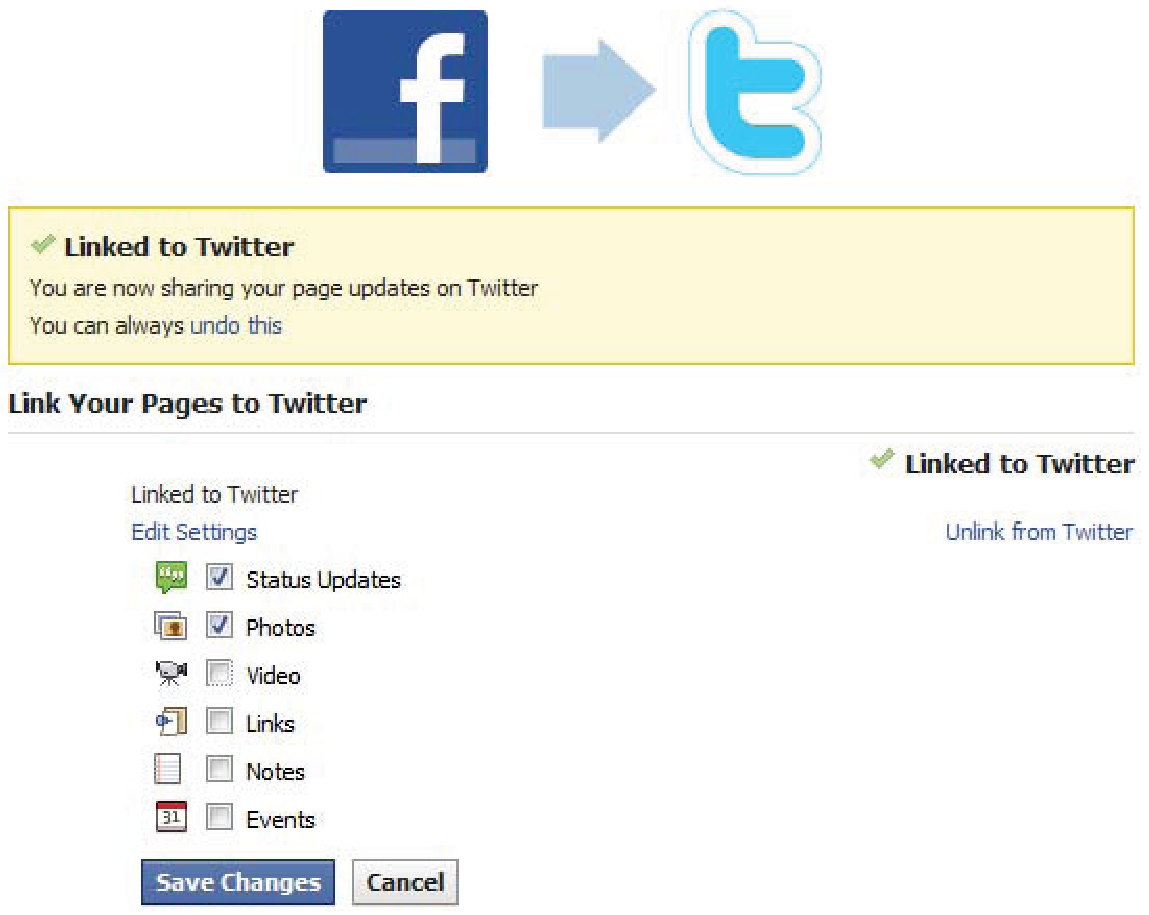}
		\label{fig:facebook2twitter}
	}
	\subfigure[Auto post from Twitter to Facebook]{
		\centering
		\includegraphics[width=0.3\columnwidth]{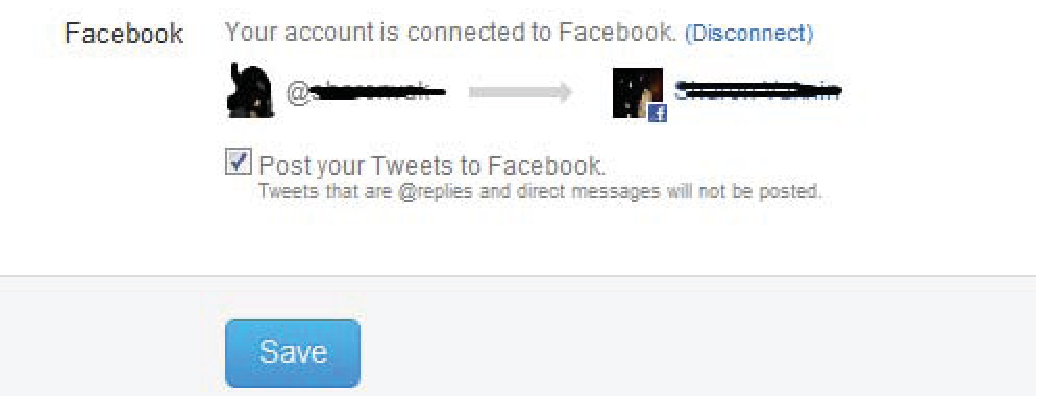}
		\label{fig:twitter2facebook}
	}\\
    \subfigure[The number of shared users between major OSNs in 2009 \cite{Anderson2009}]{
		\centering
		\includegraphics[width=0.3\columnwidth]{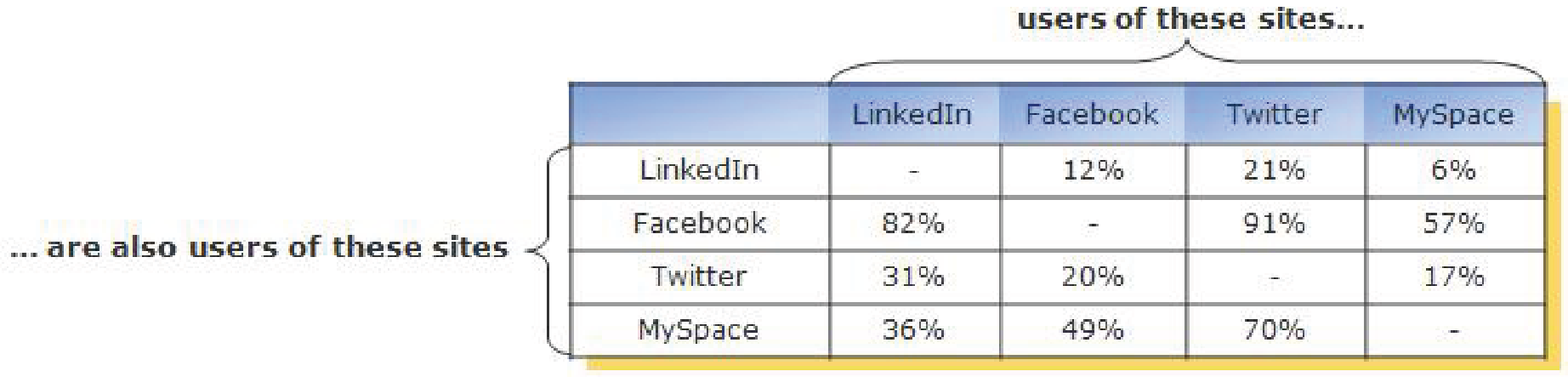}
		\label{fig:twitter2facebook}
    }
\caption{Information propagation across social networks}%
\label{fig:infoPropagation}%
\end{figure}

\vspace{5pt}

\emph{Related works.} Nearly all the existing works studied different variants of the least cost influence problem on a single network. Kempe et al. \cite{Kempe2003} first formulated the influence maximization problem which asks to find a set of $k$ users who can maximize the influence. The influence is propagated based on a stochastic process called Independent Cascade Model (IC) in which a user will influence his friends with probability proportional to the strength of their friendship. The author proved that the problem is NP-hard and proposed a greedy algorithm with approximation ratio of $(1 - 1/e)$. After that, a considerable number of works studied and designed new algorithms for the problem variants on the same or extended models such as \cite{Chen2010,huiyuanIcdcs13,yilinCIKM2012,weiliWu}. There are also works on the linear threshold (LT) model for influence propagation in which a user will adopt the new product when the total influence of his friends surpass some threshold. Dinh et al \cite{Dinh2013ToN} proved the inapproximability as well as proposed efficient algorithms for this problem on a special case of LT model. In their model, the influence between users is uniform and a user is influenced if a certain fraction $\rho$ of his friends are active.

Recently, researchers have started to explore multiplex networks with works of Yagan et al. \cite{Yagan2012} and Liu et al. \cite{Liu2012} which studied the connection between offline and online networks. The first work investigated the outbreak of information using the SIR model on random networks. The second one analyzed networks formed by online interactions and offline events. The authors focused on understanding the flow of information and network clustering but not solving the least cost influence problem. Additionally, these works did not study any specific optimization problem of viral marketing. Shen et al. \cite{yilinCIKM2012} explored the information propagation in multiplex online social networks taking into account the interest and engagement of users. The authors combined all networks into one network by representing an overlapping user as a super node. This method cannot preserve the individual networks' properties.

In this article, we studies the LCI problem which aims at finding a set of users with minimum cardinality to influence a certain fraction of users in multiplex networks. Due to the complex diffusion process in multiplex networks, it is difficult to develop the solution by directly extending previous solutions in a single network. Additionally, studying the problem in multiplex networks introduces several new challenges: (1) how to accurately evaluate the influence of overlapping users; (2) in which network, a user is easier to be influenced; (3) which network propagates the influence better. To answer above questions, we first introduce a model representation to illustrate how information propagate in multiplex networks via coupling schemes. By mapping multiple networks into one network, different coupling schemes can preserve partial or full properties of the original networks. After that, we can exploit existing solutions on a single network to solve the problem in multiplex networks. Moreover, through comprehensive experiments, we have validated the effectiveness of the coupling schemes, and also provide some interesting insights into the process of influence propagation in multiplex networks. Our main contributions are summarized as follows:

\begin{itemize}
	\item Propose a model representation via various coupling schemes to reduce the problem in multiplex networks to an equivalent problem on a single network. The proposed coupling schemes can be applied for most popular diffusion models including: linear threshold model, stochastic threshold model and independent cascading model.
	\item Provide a scalable greedy algorithm to solve the LCI problem. Especially, the improvement factor scales up with the size of the network which allows the algorithm to run on very large networks with millions of nodes.
	\item Conduct extensive experiments on both real and synthesized datasets. The results show that considering multiplex networks instead of a single network can effectively choose the most influential users.
\end{itemize}

The rest of the paper is organized as follows. In Section~\ref{se:model}, we present the influence propagation model in multiplex networks and define the problem. The lossless and lossy coupling schemes are introduced in Section~\ref{se:coupling} and Section~\ref{se:lossy}. A scalable greedy algorithm is proposed in Section~\ref{se:algorithm}. Section~\ref{se:experiment} shows the experimental results on the performance of different algorithms and coupling schemes. Finally, Section~\ref{se:conclusion} concludes the paper.

\section{Model and problem definition} 
\label{se:model}
\subsection{Graph notations}
We consider $k$ networks $G^1, G^2, \ldots, G^k$, each of which is modeled as a weighted directed graph $G^i = (V^i, E^i, \theta^i, W^i)$.
The vertex set $V^i = \{u\text{'s}\}$ represents the participation of $n^i = |V^i|$ users in the network $G^i$, and the edge set $E^i = \{(u, v)\text{'s}\}$ contains $m^i = |E^i|$ oriented connections (e.g., friendships or relationships) among network users.
$W^i = \{w^i(u, v)\text{'s}\}$ is the (normalized) weight function associated to all edges in the $i^{th}$ network.
In our model, weight $w^i(u, v)$ can also interpreted as the strength of influence (or the strength of the relationship) a user $u$ has on another user $v$ in the $i^{th}$ network.
The sets of incoming and outgoing neighbors of vertex $u$ in network $G^i$ are denoted by $N^{i-}_{u}$ and $N^{i+}_{u}$, respectively.
In addition, each user $u$ is associated with a threshold $\theta^i(u)$ indicating the persistence of his opinions. The higher $\theta^i(u)$ is, the more unlikely that $u$ will be influenced by the opinions of his friends.
Furthermore, the users that actively participate in multiple networks are referred to as \emph{overlapping users} and can be identified using methods in \cite{iofciu2011identifying, buccafurri2012discovering} (Note that identifying overlapping users is not the focus of this paper).
Those users are considered as bridge users for information propagation across networks.
Finally, we denote by $G^{1 \ldots k}$ the system consisting of $k$ networks, and by $U$ the exhaustive set of all users $U = \cup_{i = 1}^k V^i$.

\subsection{Influence Propagation Model}
We first describe the \textit{Linear Threshold} (LT) model  \cite{Dinh2013ToN}, a popular model for studying information and influence diffusion in a single network, and then discuss how LT model can be extended to cope with multiplex networks.
In the classic LT model, each node $u$ can be either \textit{active} or \textit{inactive}: $u$ is in an \emph{active} state if it is selected into the seed set, or the total influence from the in-degree neighbors exceeds its threshold $\theta(u)$, i.e, $\sum_{v \in N(u)} w(v,u) \geq \theta(u)$. Otherwise, $u$ is in an \emph{inactive} state.

In multiplex network system, given a number of $k$ networks, the information is propagated separately in each network and can only flows to other networks via the overlapping users.
The information starts to spread out from a set of seed users $S$ i.e. all users in $S$ are active and the remaining users are inactive. At time $t$, a user $u$ becomes active if the total influence from its active neighbors surpasses its threshold in some network i.e. there exists $i$ such that: 
$$\sum_{v \in N^{i-}_{u}, v \in A } {w^i(v, u)} \geq \theta^i(u)$$ 
where $A$ is the set of active users after time ($t-1$).  

In each time step, some of inactive users become activated and try to influence other users in the next time step. The process terminates until no more inactive users can be activated. If we limit the propagation time to $d$, then the process will stop after $t = d$ time steps. The set of active users in time $d$ is denoted as $A^{d}(G^{1 \ldots k}, S)$. Note that $d$ is also the number of hops up to which the influence can be propagated from the seed set, so $d$ is called the number of propagation hops.

\subsection{Problem definition}
In this paper, we address the fundamental problem of viral marketing in multiplex networks: the \textbf{Least Cost Influence} problem. The problem asks to find a seed set of minimum cardinality which influences a large fraction of users. 

\begin{definition}(Least Cost Influence (LCI) Problem)
Given a system of $k$ networks $G^{1 \ldots k}$ with the set of users $U$, a positive integer $d$, and $0 < \beta \leq 1$, the LCI problem asks to find a seed set $S \subset U$ of minimum cardinality such that the number of active users after $d$ hops according to LT model is at least $\beta$ fraction of users i.e. $|A^d(G^{1 \ldots k}, S)| \geq \beta |U|$.  
\end{definition}

When $k = 1$, we have the variant of the problem on a single network which is NP-hard to solve \cite{Chen2008}, Dinh et al. \cite{Dinh2013ToN} proved the
inapproximability and proposed an algorithm for a special case when the influence between users is uniform and a user is activated 
if a certain fraction $\rho$ of his friends are active. In the following sections, we will present different coupling strategies to reduce the problem in multiplex networks to the problem in a single network in order to utilize the algorithm design.

\section{Coupling Scheme}
A coupling scheme is an approach to project multiple networks to a single network, which can preserve important network information and reproduce the diffusion process from each individual network. Such a scheme will facilitate researchers to study various optimization problems that relate to the diffusion of information on multiple networks. In general, we can mitigate these problems to the one defined on single network and apply existing solutions to solve them. Next we specify the requirements for such schemes and the general framework.

\subsection{Coupling scheme general framework}
Our goal is to map multiple networks into a single network such that a diffusion process on multiple networks can be simulated by a process on the projected network. Two most important points are: (1) which user is active and (2) when a user is activated. Formally, a coupling scheme that maps a system of networks $G^{1 \ldots k}$ with the set of users $U$ to a network $G=(V, E)$ needs to satisfy following requirements:   

\begin{itemize}
	\item[(1)] There exists a set of nodes $\mathcal{U} \subseteq V$ and bijection function that maps users to nodes in the coupled network: $\mathcal{F}: \mathcal{U} \rightarrow U$.
	\item[(2)] There exists a time mapping function $\mathcal{T}: \mathbb{N} \rightarrow \mathbb{N}$.
	\item[(3)] User $u \in U$ is activated at time $t$ on $G^{1 \ldots k}$ iff $\mathcal{F}(u)$ is activated at time $\mathcal{T}(t)$ on $G$. 
\end{itemize} 

The first constraint reserves the identity of users in the coupled network. The second constraint allows us the know when a user is activated. The last constraint guarantees that the diffusion process is preserved, i.e., the diffusion of information on the set of user $U$ is the same on the set of nodes $\mathcal{U}$. This is the core part of the couple scheme and may be difficult to achieve. Since the main goal is to construct a solution to the studied problem on multiple networks from the solution on single network, we can relax the last condition such as $u \in U$ is activated at time $t$ on $G^{1 \ldots k}$ if $\mathcal{F}(u)$ is activated at time $\mathcal{T}(t)$ on $G$. In this case, the diffusion information is not totally reserved. The coupling scheme is called \emph{lossless coupling scheme} if the last condition is satisfied and \emph{lossy coupling scheme} otherwise. 

Since our main concern is the diffusion of information among users, such coupling scheme reserve most of the properties of the diffusion process. It helps to answer following questions:
\begin{itemize}
	\item When a node becomes active?
	\item How many nodes are activated at a specific time?
	\item Who are top influencers in the multiple networks?
\end{itemize}

Another important aspect of the coupling scheme is the activation state of nodes in $V \setminus \mathcal{U}$. In some optimization problems, the fraction of active nodes plays an important role. Thus, it is desirable for the coupling scheme to reserve the fraction of active nodes or the scale-up property. 

\begin{definition}[Scale-up Property]
A coupling scheme is said to have scale-up property if there exists a constant $c = c(G^{1 \ldots k})$ such that there is $cK$ active nodes on $G$ iff there is $K$ active  users on $G^{1 \ldots k}$.
\end{definition}



\subsection{General framework to solve some optimization problems}
With the coupling scheme, if we only consider the set of users and its mapped set on the coupled network, the diffusion process is the same on these two sets. Thus, we can design algorithms to solve various information diffusion optimization problems on multiple networks such as Influence Maximization problem \cite{Kempe2003}, Limiting the misinformation problem  \cite{budak2011limiting}, Minimum Influential Node Selection problem  \cite{zou2009latency}, etc, by the following framework: (1) Create a coupled network following a coupling scheme, (2) Use an algorithm for the studied problem on single networks to identify the set of selected nodes, (3) Use the $\mathcal{F}$ function to determine the set of selected users from the set of selected nodes on the coupled network. 
\begin{algorithm}[h]
\caption{General Framework}
\label{alg:framework}
	\begin{algorithmic}
	\State \textbf{Input:} A set of users $U$, a system of networks $G^{1 \ldots k}$, and an algorithm $\mathcal{A}$.
	\State \textbf{Output:} A solution $S \subset U$
		\State $G \gets \textit{ The coupled network of } G^{1 \ldots k}$
		\State $C \gets \textit{ Set of selected users provided by } \mathcal{A} \textit{ on } G$
		\State $S \gets \mathcal{F}(C)$
		\State Return $I$
	\end{algorithmic}
\end{algorithm}

\section{Lossless coupling schemes}
\label{se:coupling}
In this section, we present the lossless scheme to couple multiple networks into a new single network with respect to the influence diffusion process on each participant network. A notable advantage of this newly coupled graph is that we can use any existing algorithm on a single network to produce the solution in multiplex networks with the same quality.

\subsection{Clique lossless coupling scheme}

In LT-model, the first issue is solved by introducing dummy nodes for each user $u$ in networks that it does not belong to. These dummy nodes are isolated. Now the vertex set $V^i$ of $i^{th}$ network can be represented by $V^i = \{u^i_1, u^i_2, \ldots, u^i_n\}$ where $U =  \{u_1, u_2, \ldots, u_n\}$ is the set of all users. $u^i_p$ is called the \emph{representative vertex} of $u_p$ in network $G^i$. In the new representation, there is an edge from $u^i_p$ to $u^i_q$ if $u_p$ and $u_q$ are connected in $G^i$. Now we can union all $k$ networks to form a new network $G$. The approach to overcome the second challenge is to allow nodes $u^1, u^2, \ldots, u^k$ of a user $u$ to influence each other e.g. adding edge $(u^i, u^j)$ with weight $\theta(u^j)$. When $u^i$ is influenced, $u^j$ is also influenced in the next time step as they are actually a single overlapping user $u$, thus the information is transferred from network $G^i$ to $G^j$. But an emerged problem is that the information is delayed when it is transferred between two networks. Right after being activated, $u^i$ will influence its neighbors while $u^j$ needs one more time step before it starts to influence its neighbors. It would be better if both $u^i$ and $u^j$ start to influence their neighbors in the same time. For this reason, new \emph{gateway vertex} $u^0$ is added to $G$ such that both $u^i$ and $u^j$ can only influence other vertices through $u^0$. In particular, all edges $(u^i, v^i)$ ($(u^j, z^j)$) will be replaced by edges $(u^0, v^i)$ ($(u^0, z^j)$). In addition, more edges are added between $u^0$, $u^i$, and $u^j$ to let them influence each other, since the connection between gateway and representative vertices of the same user forms a clique, so we call it clique lossless coupling scheme. After forming the topology of the coupled network, we assign edge weights and vertex thresholds as following:

\emph{Vertex thresholds}. All dummy vertices and gateway vertices have the threshold of 1. Any remaining representative vertex $u^i_p$ has the same threshold as $u_p$ in $G^i$, i.e., $\theta(u^i_p) = \theta^i(u_p)$. 

\emph{Edge weights}. If there is an edge between user $u$ and $v$ in $G^i$, then the edge $(u^0, v^i)$ has weight $w(u^0, v^i) = w^i(u, v)$. The edges between gateway and representative vertices of the same user $u$ are assigned as $w(u^i, u^j)$ $= \theta(u^j)$, $\forall~~0 \leq i, j \leq k, i \neq j$ to synchronize their state together. 

A simple example of the clique lossless coupling scheme is illustrated in Fig. \ref{fig:lossless_scheme_example}.  

\begin{figure}[h!]
\centering
\subfigure[An instance of multiplex networks with 4 users. Each user is represented by vertices of the same color with different thresholds in different networks e.g. green user has thresholds of 0.3 and  0.2 in $G^2$ and $G^3$.] { 
		\centering
		\includegraphics[width=0.55\columnwidth]{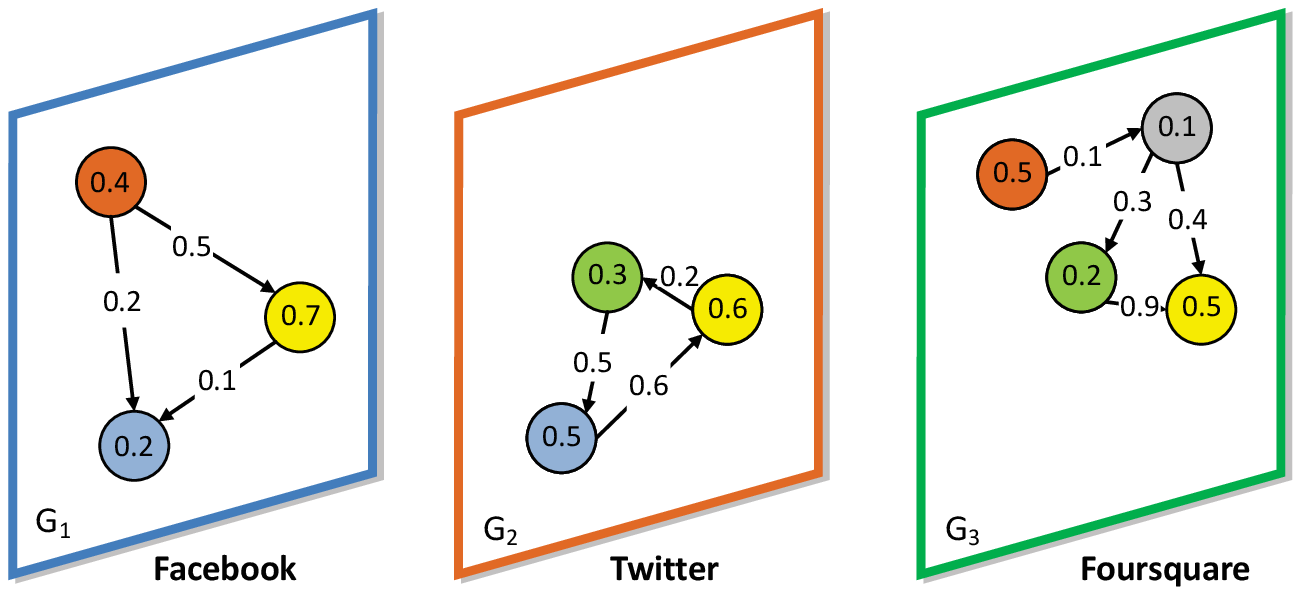}%
		\label{fig:original_network}
}
\subfigure[The influence of gateway vertices on represenative vertices represent the influence between users in multiplex networks.]  { 
		\centering
		\includegraphics[width=0.5\columnwidth]{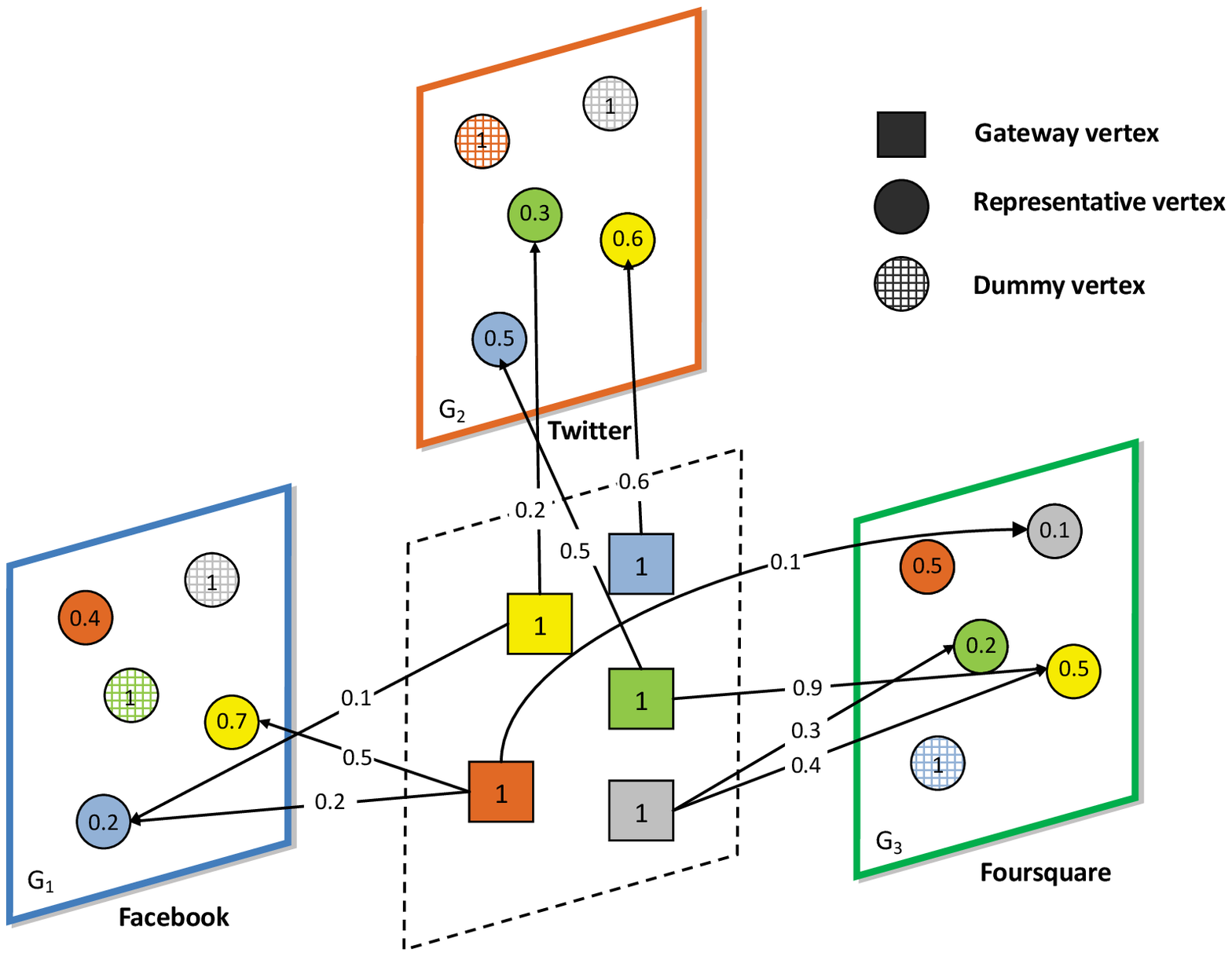}%
}
\subfigure[The connection between gateway and representative vertices of the same user.] { 
		\centering
		\includegraphics[width=0.36\columnwidth]{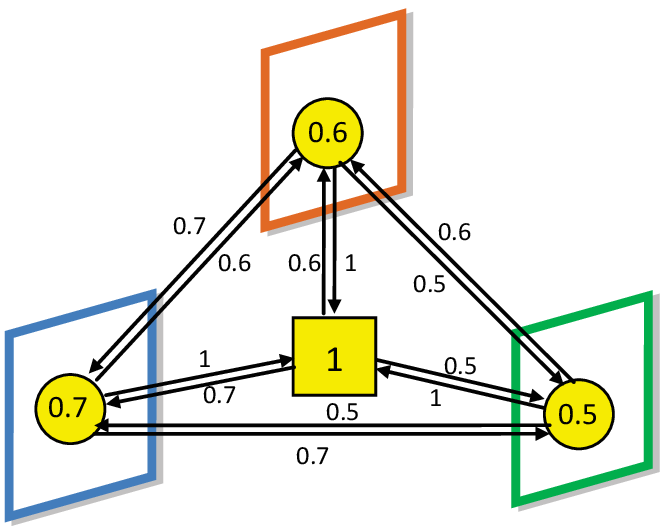}%
		\label{fig:clique}
}
\caption{An example of the \textit{clique lossless coupling scheme}.}%
\label{fig:lossless_scheme_example}%
\vspace{-15pt}
\end{figure}

Next we will show that the propagation process in the original multiplex networks and the coupled network is actually the same. Influence is alternatively propagated between gateway and representative vertices, so the problem with $d$ hops in the multiplex networks is equivalent to the problem with $2d$ hops in the coupled network.

\begin{lemma}
\label{le:lossless_activate}
Suppose that the propagation process in the coupled network $G$ starts from the seed set which contains only gateway vertices $S = \{s^0_1, \ldots, s^0_p\}$, then representative vertices are activated only at even propagation hops.
\end{lemma}

\begin{proof}
Suppose that a gateway vertex $u^0$ is the first gateway vertex that is activated at the odd hops $2d+1$. $u^0$ must be activated by some vertex $u^i$ and $u^i$ is the first activated vertex among vertices $u^1, u^2, \ldots, u^k$. It means that $u^i$ is activated in hop $2d$. Since all incoming neighbors of $u^i$ are gateway vertices, some gateway vertex becomes active in hop $2d-1$ (contradiction).
\end{proof}

\begin{lemma}
\label{le:lossless_mapping}
Suppose that the propagation process on $G^{1\ldots k}$ and $G$ starts from the same seed set $S$, then following conditions are equivalent: 
\begin{itemize}
	\item[(1)] User $u$ is active after $d$ propagation hops in $G^{1\ldots k}$. 
	\item[(2)] There exists $i$ such that $u^i$ is active after $2d-1$  propagation hops in $G$.  
	\item[(3)] Vertex $u^0$ is active after $2d$ propagation hops in $G$.
\end{itemize}

\end{lemma}

\begin{proof}
We will prove this lemma by induction. Suppose it is correct for any $1 \leq d \leq t$, we need to prove it is correct for $d = t + 1$. Denote $A^{1 \ldots k} (t)$ and $A$(t) as the set of active users and active vertices after $t$ propagation hops in $G^{1\ldots k}$ and $G$, respectively.   

$(1) \Rightarrow (2)$:
If user $u$ is active at time $t + 1$ in $G^{1\ldots k}$, it must be activated in some network $G^j$. We have:
\[\sum_{v \in N^{j-}_u \cap A^{1 \ldots k}(t) } {w^j(v, u)} \geq \theta^j(u)\] 

Due to the induction assumption, for each $v \in A^{1 \ldots k}(t)$, we also have $v^0 \in A(2t)$ in $G$. Thus:
\[
\begin{split}
\sum_{v^0 \in N^-_{u^j} \cap A(2t) } {w(v^0, u^j)} &= \sum_{v \in N^{j-}_u \cap  A^{1 \ldots k}(t) } {w^j(v, u)} \geq \theta^j(u)\\
 & = \theta(u^j) \\
\end{split}
\] 
It means that $u^j$ is active after $(2(t+1) -1)$ propagation hops. 

$(2) \Rightarrow (3)$:
If there exists $i$ such that $u^i$ is active after $2(t+1) -1$  propagation hops on $G$, then $u^i$ will activate $u^0$ in hop $2(t+1)$

$(3) \Rightarrow (1)$:
Suppose that $u^0 \notin S$ is active after $2(t+1)$ propagation hops in $G$, then there  exists $u^j$ which activates $u^0$ before. This is equivalent to:
\[
\sum_{v \in N^-_{u^j}, v \in A(2t) } {w(v, u^j)} \geq \theta(u^j)
\] 

For each $v \in A(2t)$, we also have $v \in A^{1 \ldots k}(t)$. Replace this into the above inequality we have:
\[
\begin{split}
 \sum_{v \in N^{j-}_u \cap A^{1 \ldots k}(t) } {w^j(v, u)} &= \sum_{v^0 \in N^-_{u^j} \cap A(2t) } {w(v^0, u^j)}\\
& \geq \theta(u^j) = \theta^j(u)
\end{split}
\] 

Thus, $u$ is active in network $G^j$ after $t+1$ hops.
\end{proof}

Next, we will show that the number of influenced vertices in the coupled network is $(k+1)$ times the number of influenced users in multiplex networks as stated in Theorem \ref{th:lossless_eqal}.

\begin{theorem}
\label{th:lossless_eqal}
Given a system of $k$ networks $G^{1\dots k}$ with the user set $U$, the coupled network $G$ produced by the lossless coupling scheme, and a seed set $S = \{s_1, s_2, \ldots, s_p\}$, if $A^d(G^{1\dots k}, S)$ $= \{a_1, a_2,$ $\ldots, a_q\}$ is the set of active users caused by $S$ after $d$ propagation hops in multiplex networks, then $A^{2d}(G, S)$ $ = \{a^0_1, a_1^1, \ldots, a^k_1, \ldots$, $a^0_q, a_q^1, \ldots, a_q^k\}$ is the set of active vertices caused by $S$ after $2d$ propagation hops in the coupled network.
\end{theorem}

\begin{proof}
For each user $a_i \in A^d(G^{1\dots k}, S)$ i.e. $a_i$ is active after $d$ hops in $G^{1\dots k}$, then there exists $a_i^j$ which is active after $2d-1$ hops in $G$ according to the Lemma \ref{le:lossless_mapping}. As a result, all $a^0_i, a_i^1, \ldots, a_i^k$ are active after $2d$ hops. So $B = \{a^0_1, a_1^1, \ldots, a^k_1, \ldots$, $a^0_q, a_q^1, \ldots,$ $a_q^k\}$ $\subseteq A^{2d}(G, S)$.

Let consider a vertex of $A^{2d}(G, S)$ which is:

\textit{Case 1.} A gateway vertex $u^0$ which is active after $2d$ hops in $G$, so vertex $u$ must be active after $d$ hops in $G^{1\dots k}$. This implies $u \in A^d(G^{1\dots k}, S)$, thus $u^0 \in B$.

\textit{Case 2.} An representative vertex $u^i$. If $u^i$ is active after $2d-1$ hops, then $u$ must be active after $d$ hops due to Lemma \ref{le:lossless_mapping}, thus $u \in A^d(G^{1\dots k}, S)$. Otherwise, $u^i$ is activated at hop $2d$ , it must be activated by some vertex $u^j$, $j > 0$ since all gateway vertices only change their state at even hops. Again, $u \in A^d(G^{1\dots k}, S)$. This results in $u^i \in B$.

From two above cases, we also have $A^{2d}(G, S) \subseteq B$. So that $A^{2d}(G, S) = B$, the proof is completed.
\end{proof}

Theorem \ref{th:lossless_eqal} provides the basis to derive the solution for LCI in multiplex networks from the solution on a single network. It implies an important algorithmic property of the \textit{lossless coupling scheme} regarding the relationship between the solutions of LCI in $G^{1 \ldots k}$ and $G$. The equivalence of two solutions is stated below:

\begin{theorem}
When the \textit{lossless scheme} is used, the set $S = \{s_1, s_2, \ldots, s_p\}$ influences $\beta$ fraction of users in $G^{1 \ldots k}$ after $d$ propagation hops if and only if $S^\prime = \{s^0_1, s^0_2, \ldots, s^0_p\}$ influences $\beta$ fraction of vertices in coupled network $G$ after $2d$ propagation hops.
\end{theorem} 

\textit{Size of the coupled network.} Each user $u$ has $k+1$ corresponding vertices $u^0, u^1, \ldots, u^k$ in the coupled network, thus the number of vertices is $|V| = (k+1)|U| = (k+1)n$. The number of edges equals the total number of edges from all input networks plus the number of new edges for synchronizing. Thus the total number of edges is $|E| = \sum_{i=1}^k |E^i| + nk(k+1)$.

\subsection{Star lossless coupling scheme}
In last subsection, we discussed the clique lossless coupling scheme, however, in this scheme, the number of edges to synchronize the state of vertices $u^0, u^1, \ldots, u^k$ is added up to $k(k+1)$ for each user $u$, which results in $nk(k+1)$ extra edges in the coupled network. In real networks, the number of edges is often linear to the number of vertices, while the number of extra edges greatly increases the size of the coupled network, especially when $k$ is large. Therefore, we would like to design another synchronization strategy that has less additional edges. 

\begin{figure}[h]
\centering
\includegraphics[width=0.35\columnwidth]{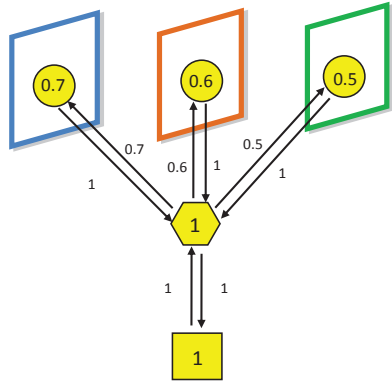}%
\caption{Synchronization of \textit{star lossless coupling scheme}.}%
\label{fig:star_lossless_scheme}%
\vspace{-5pt}
\end{figure}

Note that the large volume of extra edges is due to the direct synchronization between each pair of representative vertices of $u$ in \textit{clique lossless coupling scheme}, so we can reduce it by using indirect synchronization. In the new coupling scheme, we create one intermediate vertex $u^{k+1}$ with threshold $\theta(u^{k+1}) = 1$ and let the active state propagate from any vertex in $u^1, u^2, \ldots, u^k$ via this vertex. Specifically, the synchronization edges are established as follows: $w(u^i, u^{k+1}) = 1$ and $w(u^{k+1}, u^i) = \theta(u^i)$ $1 \leq i \leq k$;   $w(u^{k+1}, u^0) = w(u^0, u^{k+1}) = 1$. The synchronization strategy of \textit{star lossless coupling scheme} is illustrated in Fig. \ref{fig:star_lossless_scheme}. Now, the number of extra edge for each user is $2(k+1)$ and the size of the coupled network is reduced as shown in the following proposition.

\begin{proposition}
When \textit{star lossless scheme} is used, the coupled network has $|V| = (k+2)|U| = (k+2)n$ vertices and $|E| = \sum_{i=1}^k |E^i| + 2n(k+1)$ edges. 
\end{proposition}

In \textit{star lossless coupling scheme}, it takes 2 hops to synchronize the states of representative vertices of each user which leads to delaying the propagation of influence in the coupled network. Due to the similarity between \textit{star lossless scheme} and \textit{clique lossless scheme}, we state the following property of \textit{star lossless scheme} without proof.

\begin{theorem}
When \textit{star lossless coupling scheme} is used, the set $S = \{s_1, s_2, \ldots, s_p\}$ influences $\beta$ fraction of users in $G^{1 \ldots k}$ after $d$ propagation hops if and only if $S' = \{s^0_1, s^0_2, \ldots, s^0_p\}$ influences $\beta$ fraction of vertices in coupled network $G$ after $3d$ propagation hops.
\end{theorem} 
\vspace{-10pt}

\subsection{Reduced lossless schemes}

In all above coupling schemes, we create representative vertices in all networks in $G^{1 \ldots k}$ to guarantee that the number of influenced vertices in the coupled network is scaled up from the number of influenced users in the original system of networks. This creates an extraordinary redundant vertices. For example, we have $n$ users and 4 networks with $0.8n$, $0.6n$, $0.3n$, $0.2n$ users, then the total number of vertices in all network is only $1.9n$ while the number of vertices created in \textit{clique lossless scheme} and \textit{star lossless scheme} are $5n$ and $6n$, respectively. The redundant ratios of these two schemes are 260 \% and 315 \%. To overcome this redundant, we use assign weight for vertices in the coupled network and guarantee that the total weight of active vertices is scaled from the number of active users in the original system. In particular, we only create representative vertices $u^{i_1}, u^{i_2}, \ldots, u^{i_p}$ for user $u$ where $G^{i_1}, G^{i_2}, \ldots, G^{i_p}$ are networks that $u$ joins in. Each representative vertex is assigned weight 1, and the user vertex is assigned weight $k - p$. We have \textit{reduced clique lossless scheme} or \textit{reduced star lossless scheme} corresponding the method to synchronize the state of user and representative vertices is clique or star type. With this modification, the number of extra vertices is only $n$ and $2n$ when clique and star synchronization is used, respectively. The number of extra edges now depends on the participants of users.
\begin{proposition}
When \textit{reduced clique lossless scheme} or \textit{reduced star lossless scheme} is used, the coupled network has $|V| = \sum_{i=1}^k |V^i| + n$ or $|V| = \sum_{i=1}^k |V^i| + 2n$ vertices, respectively.
\end{proposition}

The relation between the a set of active vertices in coupled network and the set of active users in original networks is similar to previous schemes. We state this relation without proof below.

\begin{theorem}
When \textit{reduced clique lossless scheme} (\textit{reduced star lossless scheme}) is used, the set $S = \{s_1, s_2, \ldots, s_p\}$ influences $\beta$ fraction of users in $G^{1 \ldots k}$ after $d$ propagation hops if and only if the total weight of active vertices caused by $S' = \{s^0_1, s^0_2, \ldots, s^0_p\}$ after $2d$ ($3d$) hops in coupled network $G$ is $\beta$ fraction of the total weight of all vertices.
\end{theorem}

\subsection{Extensions to other diffusion models}
\label{se:extension}
In this section, we show that we can design lossless coupling schemes for some other well-known diffusion models in each component network. As a result, top influential users can be identified under these diffusion models. In particular, we investigate two most popular stochastic diffusion models which are Stochastic Threshold and Independent Cascading models \cite{Kempe2003}.
\begin{itemize}
    \item \textit{Stochastic Threshold model.} This model is similar to the Linear Threshold model but the threshold $\theta^i(u^i)$ of each node $u^i$ of $G^i$ is a random value in the range $[0, \Theta^i(u^i)]$. Node $u^i$ will be influenced when $\sum_{v^i \in N^-_{u^i}, v \in A } {w^i(v^i, u^i)} \geq \theta^i(u^i)$  
    \item \textit{Independent Cascading model.} In this model, there are only edge weights representing the influence between users. Once node $u^i$ of $G^i$ is influenced, it has a single chance to influence its neighbor $v^i \in N^+(u^i)$ with probability $w^i(u^i, v^i)$. 
\end{itemize}

For both models, we use the same approach of using gateway vertices, representative vertices and the synchronization edges between gateway vertices and their representative vertices. The weight of edge $(u^i, u^j)$, $0 \leq i \neq j \leq k$ will be $\Theta(u^j)$ for Stochastic Threshold model and 1 for Independent Cascading model. Once $u^i$ is influenced, $u^j$ will be influenced with probability 1 in the next time step. The proof for the equivalence of the coupling scheme is similar to ones for LT-model.

\section{Lossy coupling schemes}
\label{se:lossy}
In the preceding coupling scheme for LT-model, a complicated coupled network is produced with  large numbers of auxiliary vertices and edges. It is ideal to have a coupled network which only contain users as vertices. This network provides a compact view of the relationship between users crossing the whole system of networks. 
The loss of the information is unavoidable when we try to represent the information of multiplex networks in a compact single network.
The goal is to design a scheme that minimizes the loss as much as possible i.e. the solution for the problem in the coupled network is very close to one in the original system. Next, we present these schemes based on the following key observations.

\textit{Observation 1.} User $u$ will be activated if there exists $i$ such that:
$\sum_{v \in N^{i-}_u \cap A} {w^i(v, u)} \geq \theta^i(u)$
where $A$ is the set of active users. 
We can relax the condition to activate $u$ with positive parameters $\alpha^1(u)$, $\alpha^2(u)$, $\ldots$, $\alpha^k(u)$ as follows:

\begin{equation}
\label{coupled_condition}
\sum_{i = 1}^k \Big(\alpha^i(u) \sum_{v \in N^{i-}_u \cap A } {w^i(v, u)}\Big) \geq \sum_{i = 1}^k \alpha^i(u)\theta^i(u)
\end{equation}

Note that sometimes the condition to activate $u$ is met, but the condition (\ref{coupled_condition}) still needs more influence from $u$'s friends to satisfy. The more this  need for extra influence is, the looser condition (\ref{coupled_condition}) is. We can reduce this redundancy by increasing the value of $\alpha^i(u)$ proportional to the value of $\sum_{v \in N^{i-}_u \cap A } {w^i(v, u)} - \theta^i(u)$. In the special case, if $\sum_{v \in N^{i-}_u \cap A } {w^i(v, u)} > \theta^i(u)$ and we choose $\alpha^i(u) \gg \alpha^j(u)$, $\forall j \neq i$, then there is no redundancy. Unfortunately, we do not know before hand in which network user $u$ will be activated, so we can only choose parameters heuristically.

\textit{Observation 2.} When user $u$ participates in multiple networks, it is easier to influence $u$ in some network than the others. The following simple case illustrate such situation. Suppose that we have two networks. In network 1, $\theta^1(u) = 0.1$ and $u$  has 8 in-neighbors, each neighbor $v$ influences $u$ with $w^1(v, u) = 0.1$. In network 2, $\theta^2(u) = 0.7$ and $u$  has 8 in-neighbors, each neighbor $v$ influences $u$ with $w^2(v, u) = 0.1$. The number of active neighbors to activate $u$ is 1 and 7 in network 1 and 2, respectively. 

\emph{Easiness.} Intuitively, we can say that $u$ is easier to be influenced in the first network. We quantify the \textit{easiness} $\epsilon^i(u)$ that $u$ is influenced in network $i$ as the ratio between the total influence from friends and the threshold to be influenced: $\epsilon^i(u) = \frac{\sum_{v \in N^{i-}_u} w^i(v, u)}{\theta^i(u)}$. 
We can use the easiness of a user in networks as the parameters of the condition \ref{coupled_condition}. 

\begin{figure}%
\centering
\includegraphics[width=0.3\columnwidth]{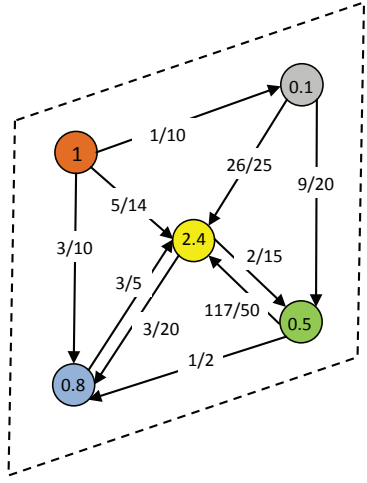}%
\caption{Lossy coupled network using easiness parameters.}%
\label{fig:lossy}%
\vspace{-17pt}
\end{figure}

Based on above observations, we couple multiplex networks into one using parameters $\{\alpha^i(u)\}$. The vertex set is the set of users $V = \{u_1, u_2, \ldots, u_n\}$. The threshold of vertex $u$ is set to 
$\theta(u) = \sum_{i = 1}^k \alpha^i(u)\theta^i(u)$

The weight of the edge $(v, u)$ is: $w(v, u) = \sum_{i = 1}^k {\alpha^i(u) w^i(v, u)}$
where $w^i(v, u) = 0$ if there is no edge from $v$ to $u$ in $i^{th}$ network. 

Then the set of edges is $E = \{(v, u) | w(v, u) > 0\}$. Fig. \ref{fig:lossy} illustrates the loosy coupled network of networks in Fig. \ref{fig:lossless_scheme_example}.

Besides easiness, other metrics can be used for the same purpose. We enumerate here some other metrics.

\textit{Involvement.} If a user is surrounded by a group of friends who have high influence on each other, he tends to be influenced. When a few of his friends are influenced, the whole group involving him is likely to be influenced. We estimate \textit{involvement} of a node $v$ in a network $G^i$ by measuring how strongly the 1-hop neighborhood $v$ is connected and to what extent influence can propagate from one node to another in the 1-hop neighborhood. Formally we can define \textit{involvement} of a node $v$ in network $G^i$ as:
$\sigma_{v}^{i}= \sum_{x,y \in  N^{i}_v \cup \{v\} } \frac{{w^{i}(x,y)}} {\theta_{y}^{i}}$
where $N^{i}_v = N^{i+}_v \cup N^{i-}_v$ is the set of all neighbors of $v$ (both in-coming and out-going).

\textit{Average.} All parameters have the same value $\alpha^i(u) = 1$

Next we show the relationship between the solution for LCI in the lossy coupled network and the original system of networks. As discussed in the above observations, if the propagation process starts from the same set of users in $G^{1 \ldots k}$ and the coupled network $G$, then the active state of a user in $G$ implies its active state in $G^{1 \ldots k}$. It means that if the set of users $S$ activates $\beta$ fraction of users in $G$, it also activates at least $\beta$ fraction of users in $G^{1\ldots k}$. It implies that if a seed set is a feasible solution in $G$, it is also a feasible solution in $G^{1\ldots k}$. Thus we have the following result.
\begin{theorem}
When the lossy coupling scheme is used, if the set of users $S$ activates $\beta$ fraction of users in $G$, then it activates at least $\beta$ fraction of users in $G^{1 \ldots k}$.
\end{theorem} 

\section{Algorithms}
\label{se:algorithm}
In this section, we describe a greedy algorithm and its improvement in terms of scalability in large networks. In the state of the art work, Dinh et al. \cite{Dinh2013ToN} only solved the problem in a special case where the threshold is uniform and the required fraction of active nodes is the same for all nodes.

\subsection{Improved greedy algorithm}

\begin{algorithm}[tbh]
\caption{Improved Greedy}
\label{alg:improved_greedy_MIP}
	\begin{algorithmic}
	\State \hspace{-0.1in} \textbf{Input:} A system of networks $G^{1 \ldots k}$, fraction $\beta$, $T$, $R$.
	\State \hspace{-0.1in} \textbf{Output:} A small seeding set $S$
	\State $G \gets \textit{ The coupled network of } G^{1 \ldots k}$		
		\State $C \gets \textit{ Set of user vertices}$
		\State $I \gets \emptyset$, $Counter \gets 0$
		\State Initialize a heap: $H \gets \emptyset$ 
		\For {$u \in C$}
			\State $H.push((u, f_{\emptyset}(u)))$
		\EndFor
		
		\While {Number of active vertices $\leq \beta|V|$}
			\State $Counter \gets Counter + 1$
			\If {$Counter \text{ \% } R == 0$}
				\State Update key values of all elements in $H$
			\Else
				\State $A \gets \emptyset$
				\For {$i = 1 \text{ to } T$}
					\State $(u, f(u)) \gets H.extract{-}max()$ 
				\EndFor
				\For {$u \in A$}
					\State $H.push((u, f_I(u))) $ 
				\EndFor
			\EndIf
			\State $(u, f(u)) \gets H.extract{-}max()$
			\State $I \gets I \cup \{u\}$			
		\EndWhile
		\State $S \gets \text{ corresponded users } G^{1 \ldots k} \text{ of nodes in } I$
		\State \textbf{Return} $S$
	\end{algorithmic}
\end{algorithm}

The bottle neck of the native greedy algorithm is to identify the best node to be selected in each iteration, thus we focus on reducing the evaluating the computational cost of this step while maintaining the same quality of selected nodes. We notice that the marginal gain function $f_I(\cdot)$ is recomputed for all unselected nodes and it does not change much after a single iteration. Since we only select the best one, which indicates that nodes with higher marginal gain in previous round are necessary to be reevaluated. Therefore, we use a max heap to store the marginal gain and extra the top one to reevaluate it, if it is not of largest marginal value, it will be pushed back. Since the time to extract/push an element to the heap is $O(\log n)$, the total computation cost for each iteration is $O(T(m + n) + T\log n) = O (T(m+n))$. Normally, $T$ is much smaller than $n$, so the running time is improved significantly. In addition, due to the property of the Linear Threshold model, the required influence (remained threshold value after subtracting the influence of activate neighbors) to activate a node is decreasing when the seed set grows up. When a large number of nodes is selected, there are many nodes are very easy to be activated. Thus the marginal gain of a node can accumulate to a large value. If we just apply the proposed strategy, we will never evaluate these nodes again. Therefore, we need to do exhaustively reevaluation periodically. Combine two strategies together, we present the idea of using heavy and light iterations alternatively. In heavy iteration, we will update $f_I(\cdot)$ of all unselected nodes while only top $T$ nodes are reevaluated in light iteration. Since we do not want too many heavy iterations, we only use one such iteration after every $R$ iterations. With this implementation, the running time is reduced much while the quality of the solution only fluctuates infinitesimal. The improved Greedy is described in Algorithm \ref{alg:improved_greedy_MIP}.

This algorithm will terminate when the number of influenced users is larger then the required fraction of total users. The complexity of this algorithm is $O((m+n)\cdot nd)$ in the worst case scenario, however, after applying the above discussed update techniques, the running time can be improved up to 700 times faster than the native greedy from experimental results.

\subsection{Comparing to Optimal Seeding}
In this section, we evaluate the performance of proposed algorithm with different coupling schemes to the optimal solution. We formulate the LCI problem to a 0-1 Integer Linear Programming (ILP) problem as follows.

{\begin{align}
    \mbox{minimize}    & \displaystyle\sum_{v \in V}  x_v^0  \nonumber \\ 
    \mbox{subject to}  & \displaystyle\sum_{v \in V} x_v^d \geq \beta |V|      \nonumber \\
            & \displaystyle\sum_{w \in N(v)} x_u^{i-1}w_{uv} +  \theta_v\cdot x_v^{i-1} \geq \theta_v \cdot x_v^i \nonumber \\
                       &   \makebox[0.97in]{ }\forall v \in V, i=1..d \nonumber \\
                & x_v^i     \geq x_v^{i-1}\makebox[0.39in]{ }\forall v \in V, i=1..d \nonumber \\
                & x_v^i \in \{0, 1\}\makebox[0.35in]{ }  \forall v \in V, i=0..d \nonumber
\end{align}
}
where $x_v^i = 1$ if $v$ is active in round $i$, otherwise, $x_v^i = 0$.

Since solving IP is NP-hard, we can not run the IP on large networks. Moreover, we have to run the IP on the coupled network with clique lossless coupling scheme, where there will be additional nodes and edges created, and the propagation hops need to be doubled (Theorem 2). Therefore, to evaluate the performance of our algorithm, we compare the result with small size synthesize networks. For generating networks, 50 nodes are randomly chosen from a 100 users base, and the probability of connecting each pair of nodes is $p=0.04$, which yields a coupled network of 300 users and with an expected average degree 2. Fig.4(a) shows the obtained seed size with influence fraction from 0.2 to 0.8 under all coupling schemes. And we also evaluate the impact of setting different propagation hops on seed size in Fig.4(b) with influence fraction $\beta=0.4$. 

\begin{figure}[h]	
  \centering
	\subfigure[] {
	\includegraphics[width=0.4\columnwidth]{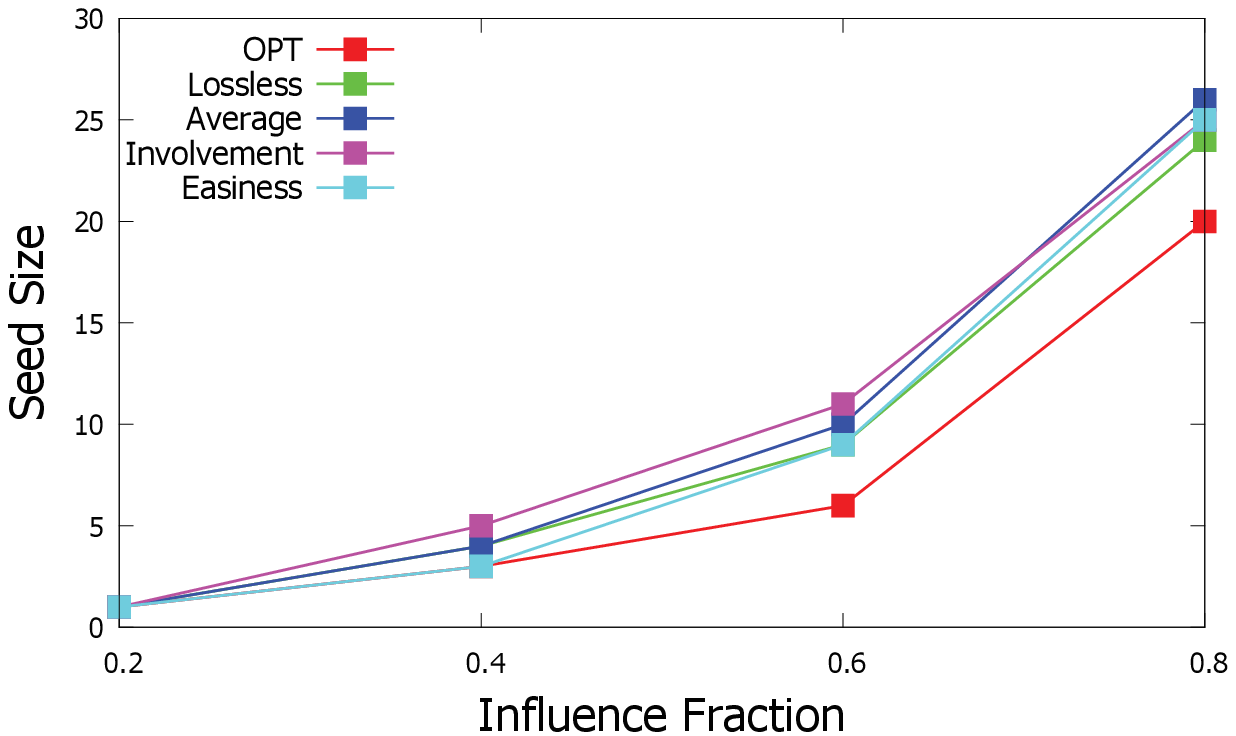}
	}
	\subfigure []{
	\includegraphics[width=0.4\columnwidth]{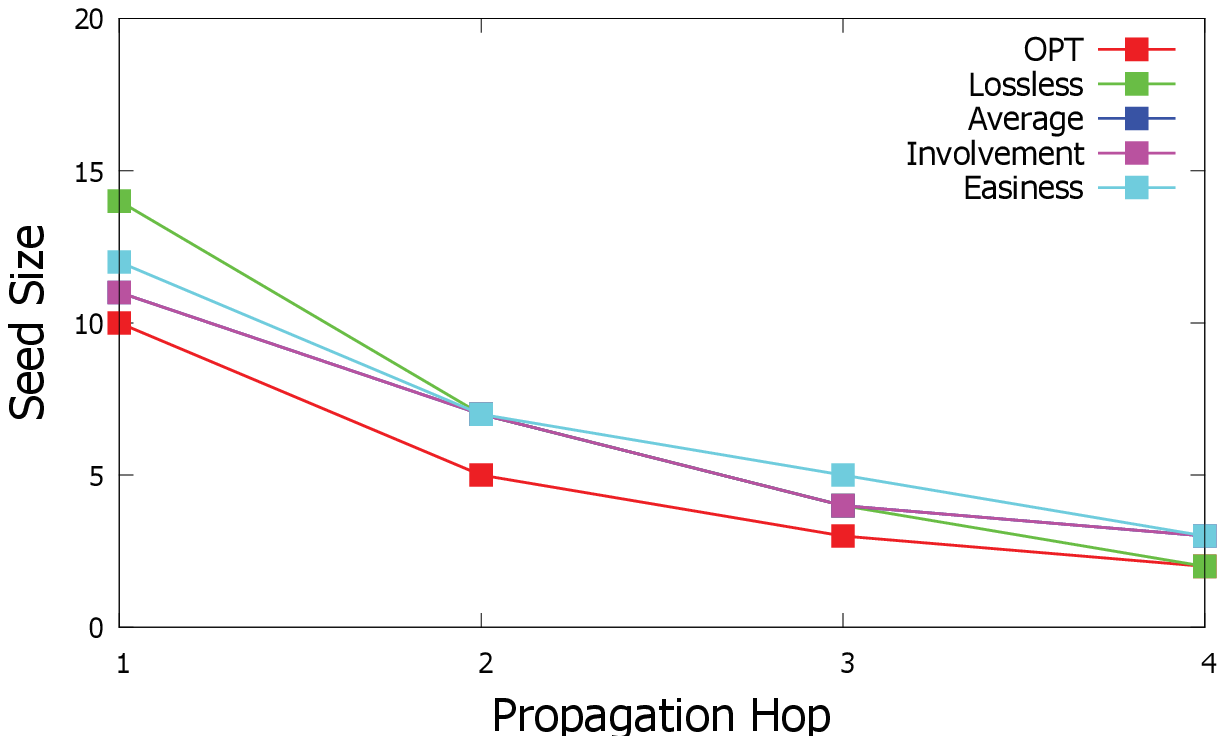}
	}	
\caption{Seed size on synthesize network.}
\label{fig:optimal}
\end{figure}

The optimal seeding along with the results of the improved greedy are shown in Fig.4. As can be seen in Fig. 4(a), the seeding sizes obtained from the proposed algorithm are close to the optimal solution while varying the influenced fraction $\beta$. The same phenomenon is also shown by varying the number of propagation hops in Fig.4(b). Especially, when the number of hops is relatively larger, the result is only one or two more than the optimal solution.

\section{Experiments}
\label{se:experiment}

In this section, we show the experimental results to compare the proposed coupling schemes and utilize these coupling schemes to analyze the influence diffusion in multiplex networks. 
First, we compare lossless and lossy coupling schemes to measure the trade-off between the running time and the quality of solutions.  In particular, for those different kinds of lossless coupling methods, all of them can preserve complete information of all networks. As a result, the quality of seeds are the same, the only difference would be the running time which have been theoretically proved in Proposition 1 and Proposition 2. Therefore, we only chose the clique coupling scheme to be evaluated. Second, we investigate the relationship between networks in the information diffusion to address the following questions: (1) What is the role of overlapping users in diffusing the information?  (2) What do we miss when considering each network separately? (3) How and to what extent does the diffusion on one network provide a burst of information in other networks?

\subsection{Datasets}
\emph{Real networks}. We perform experiments on two datasets:
\begin{itemize}
  \item \emph{Foursquare} (\emph{FSQ}) and \emph{Twitter} networks \cite{yilinCIKM2012} 
  \item Co-author networks in the area of Condensed Matter(CM) \cite{newman2001structure},  High-Energy Theory(Het) \cite{newman2001structure}, and  Network Science (NetS) \cite{newman2006finding}.
\end{itemize} 

The statistics of those networks are described in Table \ref{tab:DataSet}.  The number of overlapping users in the first dataset FSQ-Twitter is 4100 \cite{yilinCIKM2012}. We examine the in and out-degree of overlapping nodes. For the second dataset, we match overlapping users based on authors' names. The numbers of overlapping users of the network pairs  CM-Het, CM-NetS, and Het-NetS are 2860, 517, and 90, respectively. 
While the edge weights are provided for co-author networks, only the topology is available for Twitter and Foursquare networks. Thus to assign the weight of each edge, we adopt the method in \cite{Kempe2003}, where each edge weight is randomly picked from 0 to 1 and then normalize it so that the total weight of in-coming edges is sum up to 1 for each node. This is suitable since the influence of user $u$ on user $v$ tends to be small if $v$ is under the influence of many friends. Finally, we also adopt the assignment of threshold in \cite{Kempe2003} where all thresholds are randomly chosen from 0 to 1.

\begin{table}[h]
\caption{Datasets description}
\centering
\begin{tabular}{l c c c}
\hline\hline
\textbf{Networks}        & \textbf{\#Nodes} & \textbf{\#Edges}    & \textbf{Avg. Degree} \vspace{3pt} \\
\hline  
Twitter    & 48277 & 16304712      & 289.7 \\
        
FSQ & 44992 & 1664402  &       35.99 \\
        
CM & 40420 & 175692        & 8.69 \\
				
Het & 8360 & 15751      & 1.88 \\
		
NetS & 1588 & 2742      & 1.73 \\
\hline
\end{tabular}	
\label{tab:DataSet}		
\end{table}

\emph{Synthesized networks}. We also use synthesized networks generated by Erdos-Renyi random network model \cite{erdos1960evolution} to test on networks with controlled parameters. There are two networks with 10000 nodes which are formed by randomly connecting each pair of nodes with probability $p_1 = 0.0008$ and $p_2 = 0.006$. The average degrees, 8 and 60, reflect the diversity of network densities in reality. Then, we select randomly $f$ fraction of nodes in two networks as overlapping nodes.  We shall refer to $0\leq f\leq 1$ as the \emph{overlapping fraction}. The edge weights and node thresholds are assigned as Twitter.

\emph{Setup.} We ran all our experiments on a desktop with an Intel(R) Xeon(R) W350 CPU and 12 GB RAM. The number of hops is $d = 4$ and the \emph{influenced fraction} $\beta = 0.8$, unless otherwise mentioned.

\subsection{Comparison of coupling schemes}
We evaluate the impact of the coupling schemes on the running time and the solution quality of the greedy algorithm to solve the LCI problem.

\emph{Solution quality.} 
As shown in Figs. \ref{fig:cc_01} and \ref{fig:cc_03}, the greedy algorithm provides larger seed sets but runs faster in lossy coupled networks than lossless coupled networks. 
In both Twitter-FSQ and the co-author networks, the seed size is smallest when the lossless coupling scheme is used. 
It is as expected since the lossless coupling scheme reserves all the influence information which is exploited later to solve LCI. 
However, the seed sizes are only a bit larger using the lossy coupling schemes.
In the lossy coupling schemes, the information is only lost at overlapping users which occupy a small fraction the total number of users (roughly 5\% in FSQ-Twitter and 7\% in co-author networks). 
Thus, the impact of the lossy coupling schemes on the solution quality is small especially when seed sets are large to influence a large fraction of users.

A closer examination reveals the relative effectiveness of the coupling methods on the seed size. That is when the seed size is significantly small, the lossess coupling outperforms all the lossy methods. For example, when the overlapping fraction $f=0.8$, the solution using the lossless coupling is roughly 55\% of that in the solution using the (lossy) Easiness, and the solution using Easiness is about 15\% smaller than the other two lossy methods (Fig. \ref{fig:cc_7}).

\emph{Running time.} The greedy algorithm runs much faster in the lossy coupled networks than the lossless ones in general. 
As shown in Figs. \ref{fig:cc_02} and \ref{fig:cc_04}, using the lossy coupling reduces the running times by a factor of 2 in FSQ-Twitter and a factor 4 in the co-author networks compared to the lossless coupling. 
The major disadvantages of the lossless coupling scheme are the redundant nodes and edges. 
Therefore, the lossy coupling schemes works better on networks that are sparse and the number of overlapping users is small. 

\begin{figure}[!h]
  \centering
	\subfigure[\scriptsize{Seed size - co-author networks}]{
	\label{fig:cc_01}	
  \includegraphics[width=0.4\columnwidth]{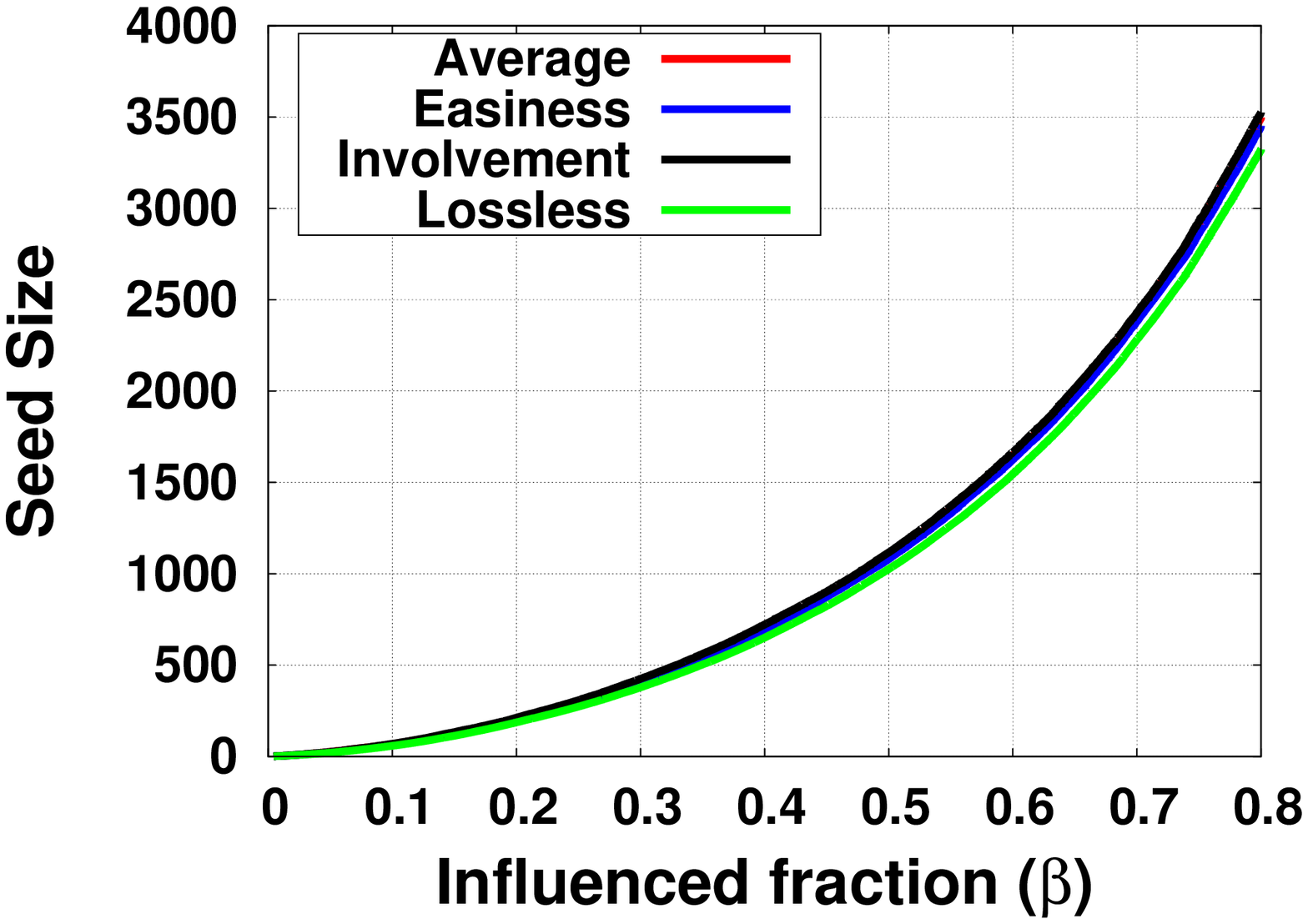}
	}
	\subfigure[\scriptsize{Seed size - FSQ-Twitter}]{
  \includegraphics[width=0.4\columnwidth]{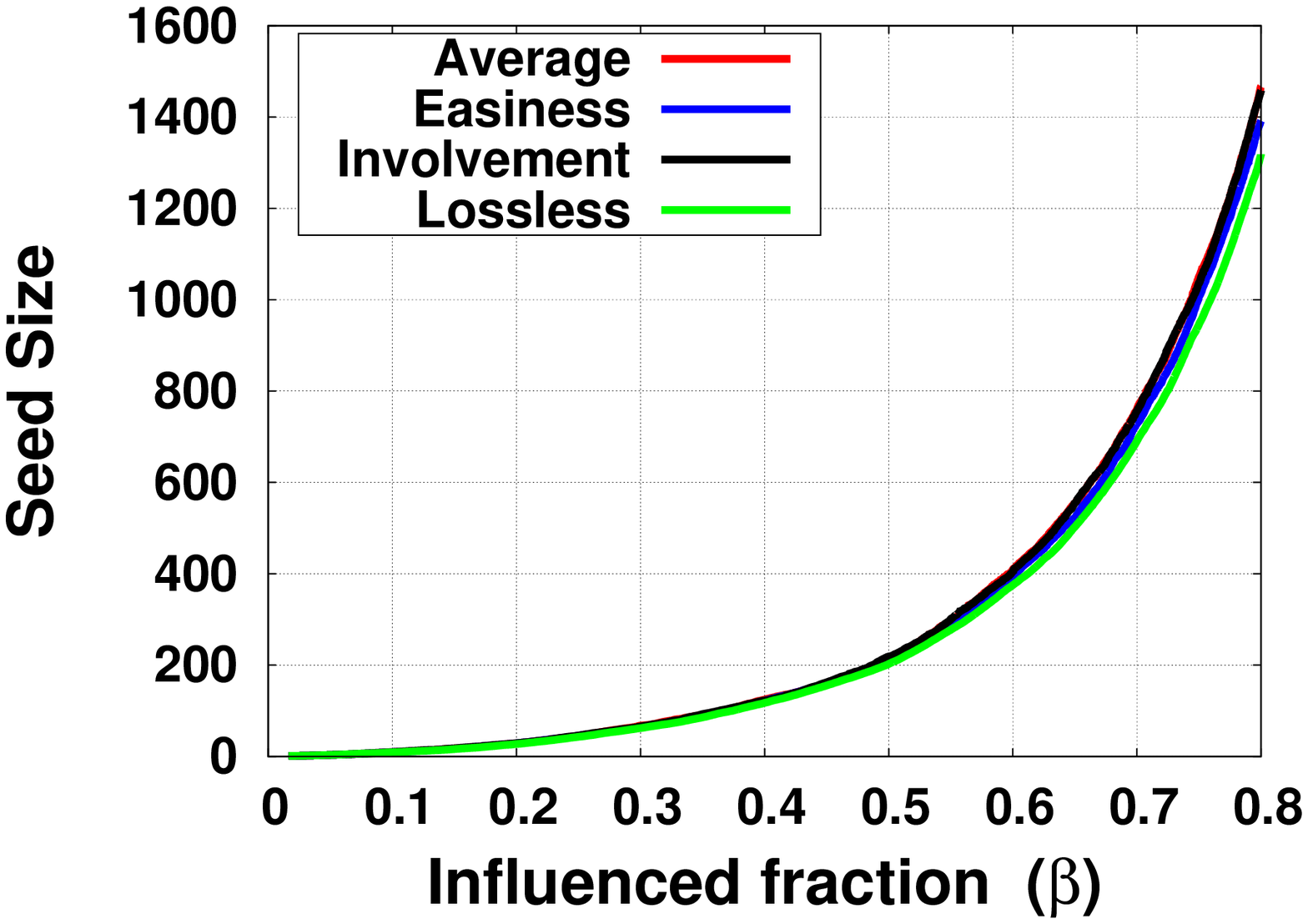}
  \label{fig:cc_03}
	}
	\subfigure[{\scriptsize Running time - co-author networks}]{
	\includegraphics[width=0.4\columnwidth]{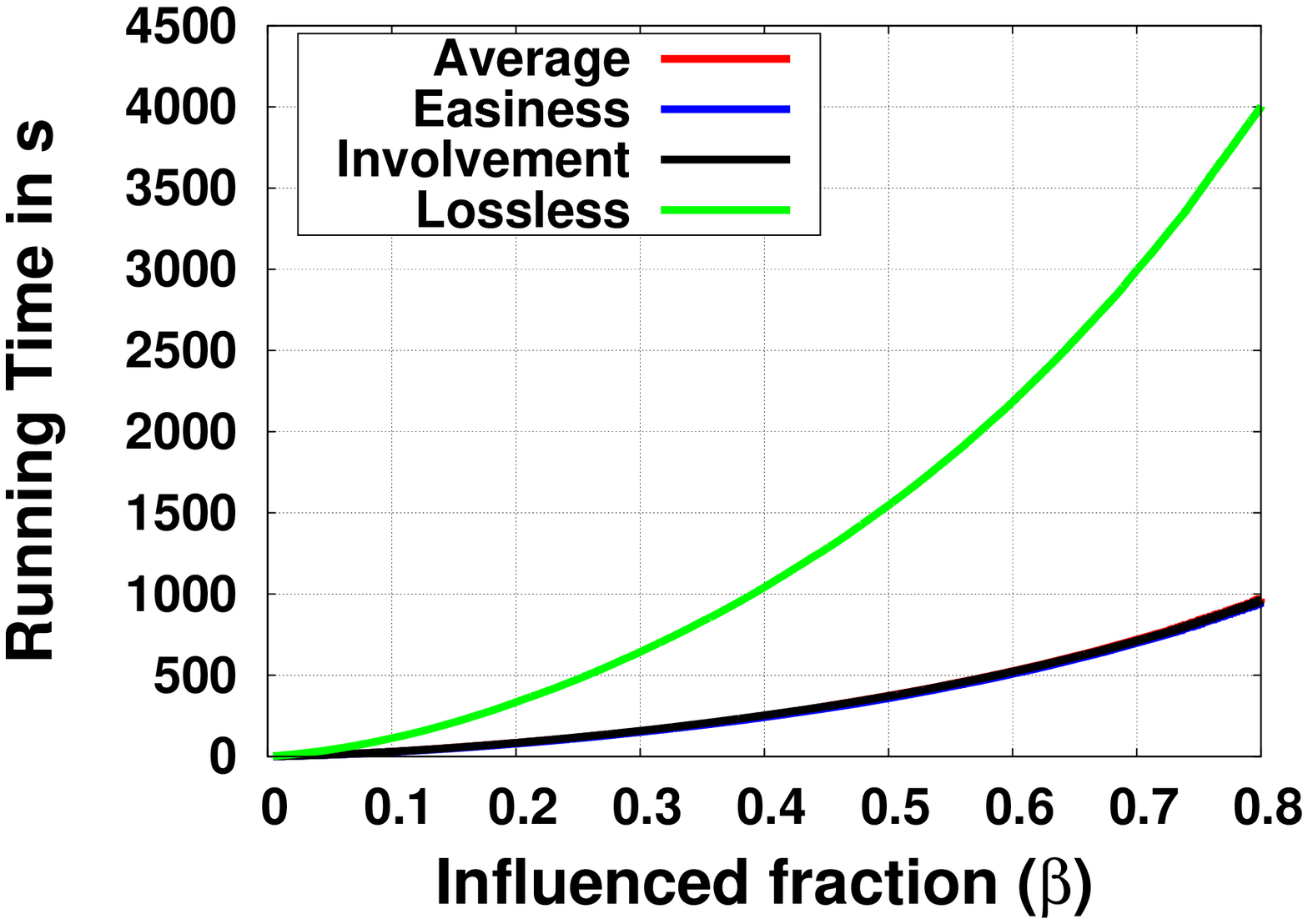}
	\label{fig:cc_02}
	}
	\subfigure[{\scriptsize Running time - FSQ-Twitter}]{
	\includegraphics[width=0.4\columnwidth]{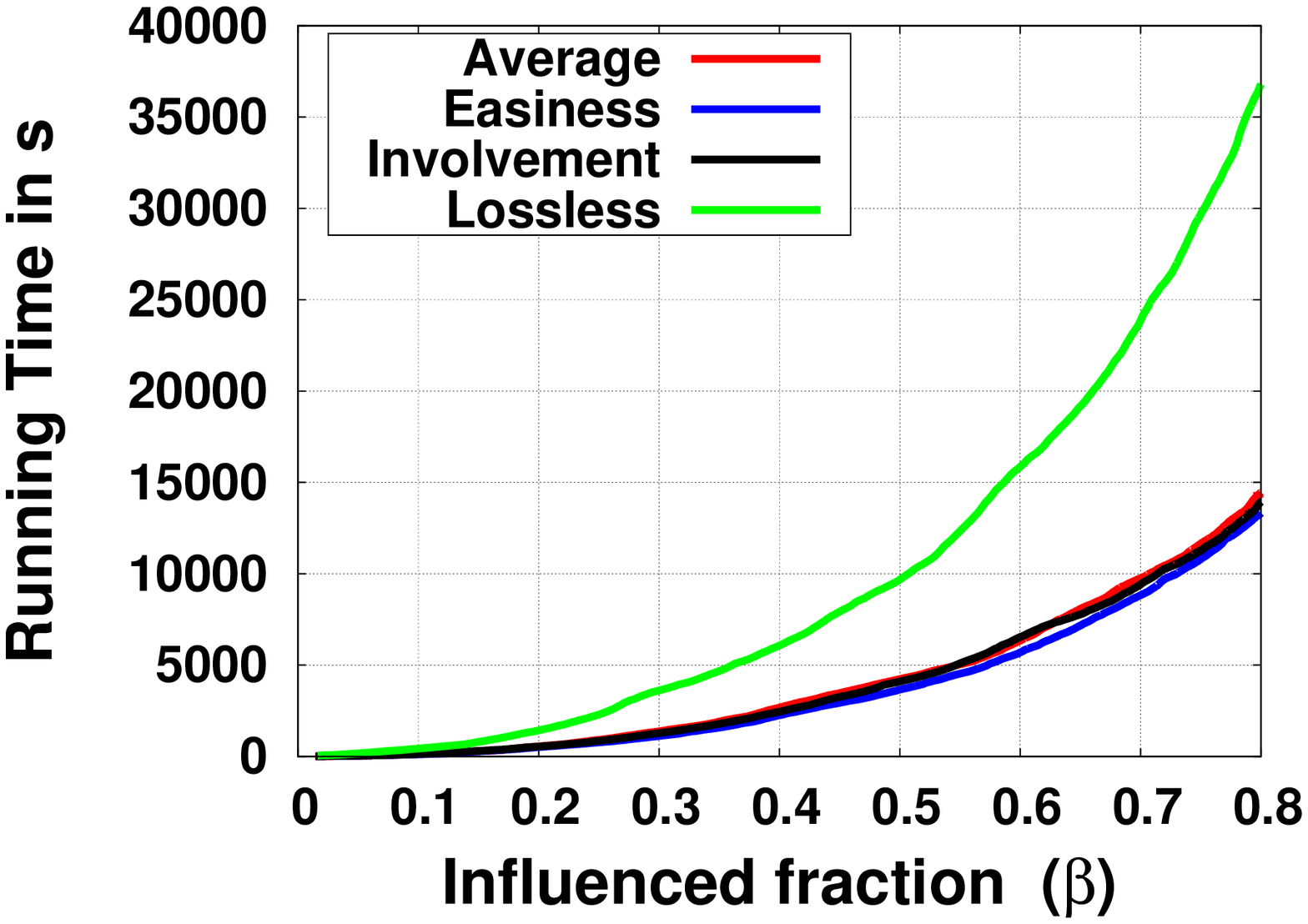}
	\label{fig:coupling_compare_run}
	\label{fig:cc_04}
	}	
\caption{Impact of coupling schemes on finding the minimum seed set}
\label{fig:coupling_compare}
\end{figure}

However, in some other cases, the lossless coupling scheme is more efficient. As shown in Fig. \ref{fig:cc_8}, the running time in the lossless coupled networks is larger in the beginning, but gradually reduces down and beats other methods at $f = 0.4$. The larger $f$ is, the larger the ratio between seed size in the lossless and lossy coupled networks is. As the running time depends on the seed size, thus it reduces faster in the lossless coupled network with larger overlapping ratio.

\emph{Overall}, the lossless coupling scheme returns solution with higher quality, especially when the seed set is small. However, if the constraint of running time and the memory are of priority, the lossy Easiness coupling scheme offers an attractive alternative.     

\subsection{Advantages of using coupled networks.} 
To understand the benefit of taking consideration of overlapping users and coupled network, in this part, we are going to compare the seed size with/out using coupled network. In particular, we do two comparisons on: 1) influencing a fraction $\beta$ of the nodes in \emph{all networks} by selecting seeds from each network and taking the union to compare with seeds achieved from lossless coupling scheme; 2) influencing a fraction $\beta$ of the nodes in \emph{a particular network} by only choosing seeds from that network compare to the seeds obtained from lossless coupling scheme.  

\begin{figure}[H]
\centering
\subfigure[Seed size] {
  \includegraphics[width=0.4\columnwidth]{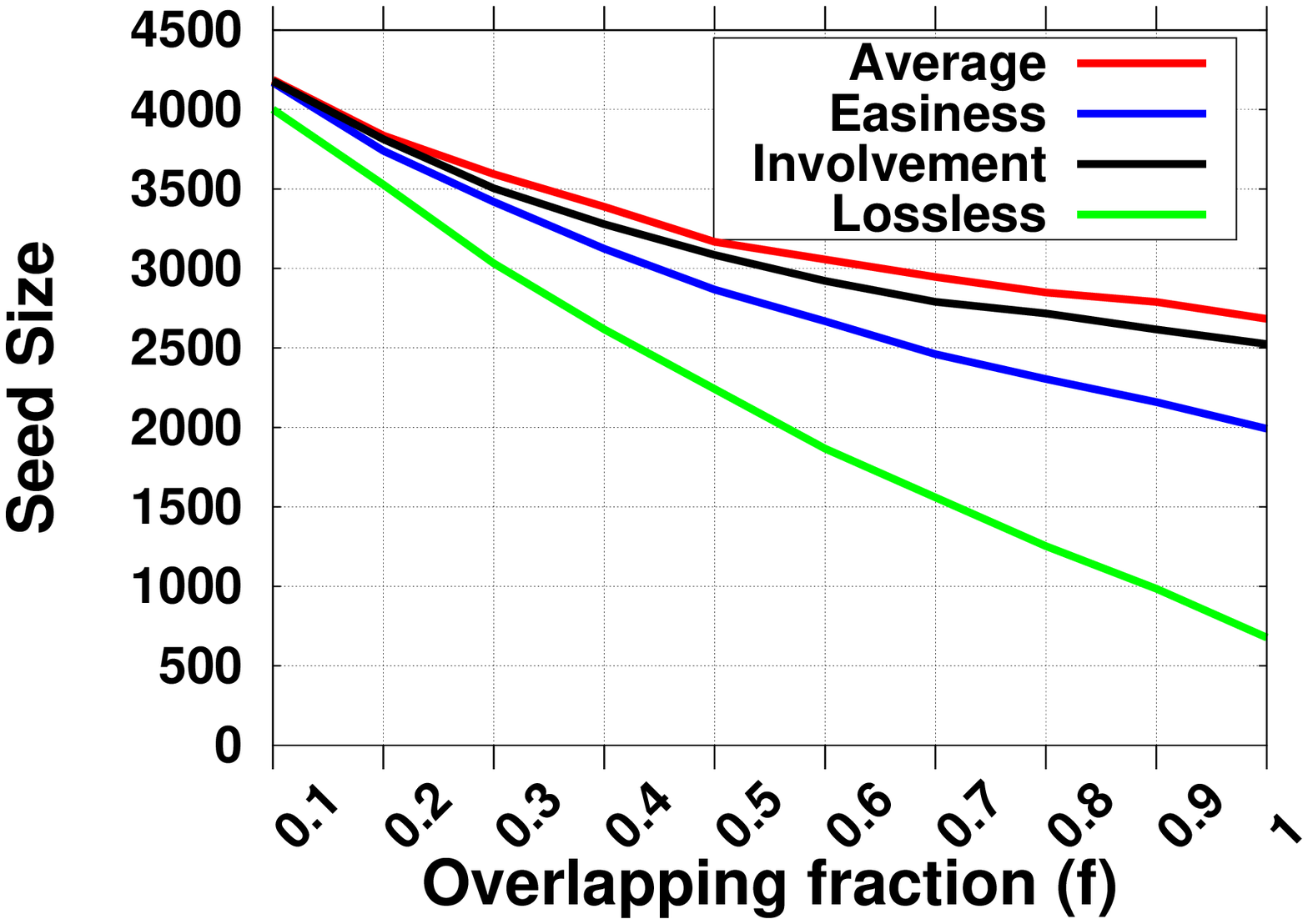}
  \label{fig:cc_7}
	}
	\subfigure[Running time]{
	\includegraphics[width=0.4\columnwidth]{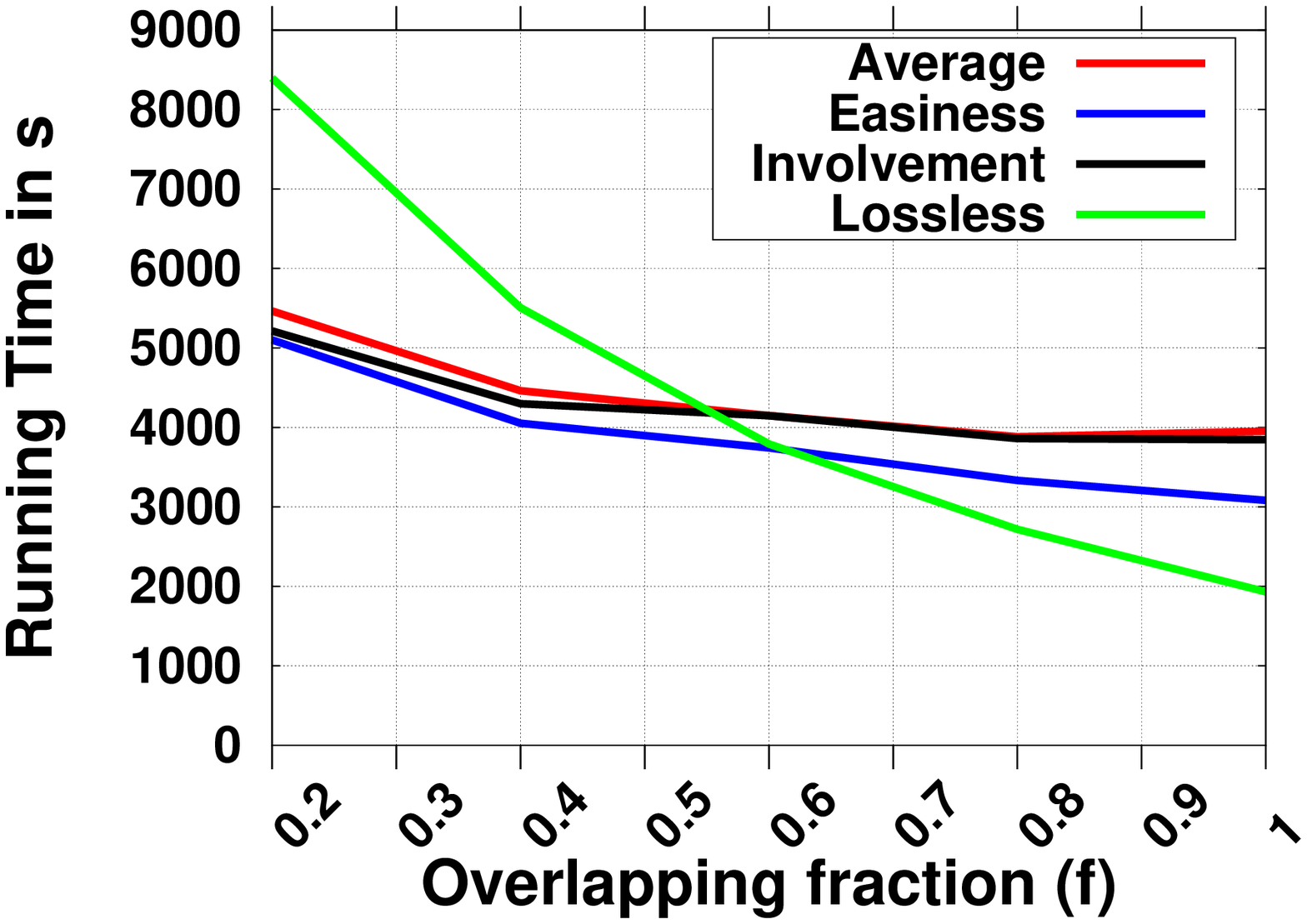}
	\label{fig:cc_8}
	}	
\caption{Comparing coupling schemes in the synthesized networks}
\label{fig:coupling_compare_overlap}
\end{figure}
 
The results for the first scenario are shown in Fig. \ref{fig:Individual_VS_Coupled}. The seed obtained by the lossless coupling method outperforms other methods. The size of the union set is approximately 30\% and 47\% larger than lossless coupling method in co-author and FSQ-Twitter, respectively. This shows that overlapping users do propagate information through several networks and thus effectively help reduce the overall seed size. 

\begin{figure}[h]%
  \centering
	\subfigure[Co-author networks] {
	\includegraphics[width=0.4\columnwidth]{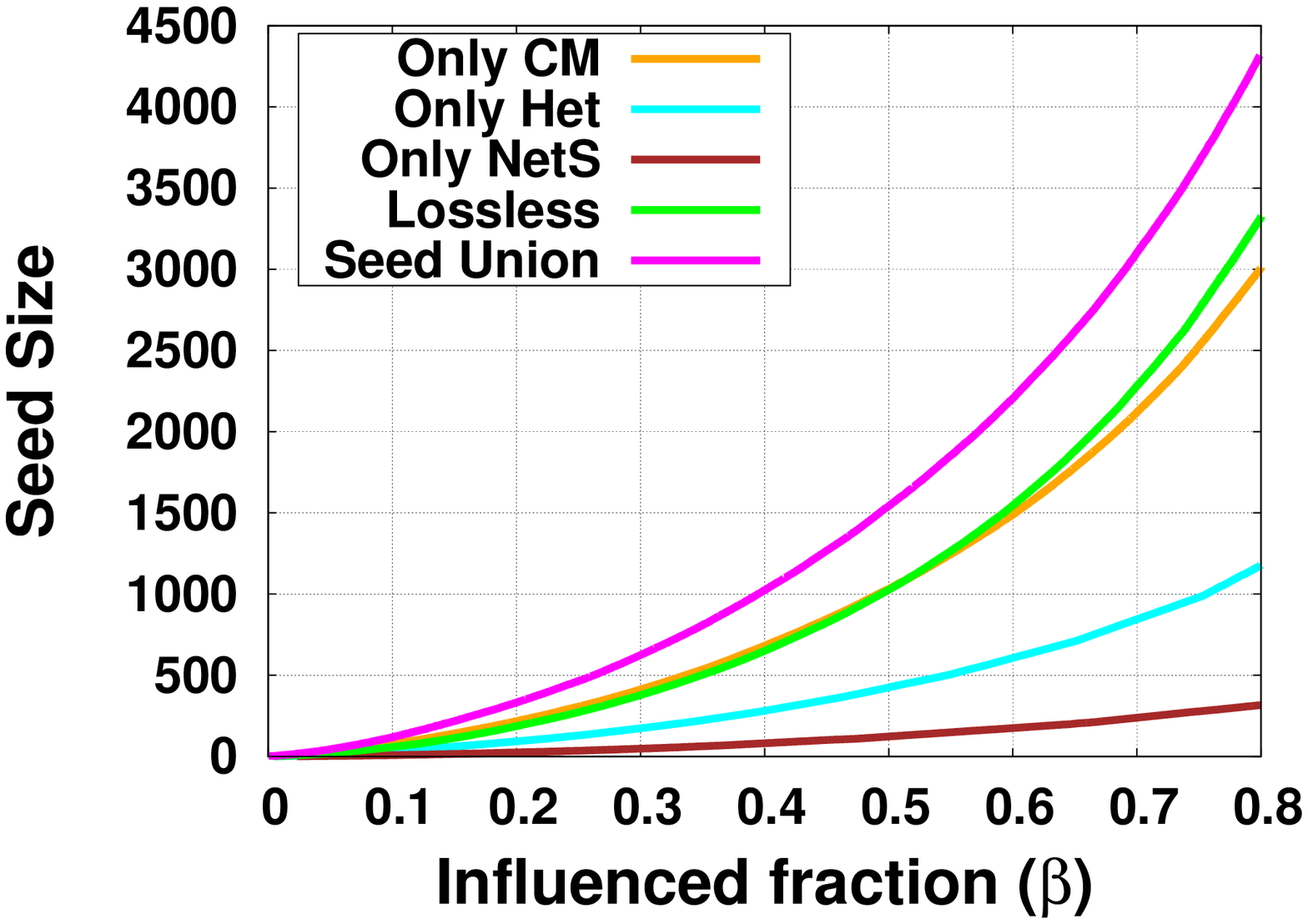}
	}
	\subfigure[FSQ and Twitter] {
	\includegraphics[width=0.4\columnwidth]{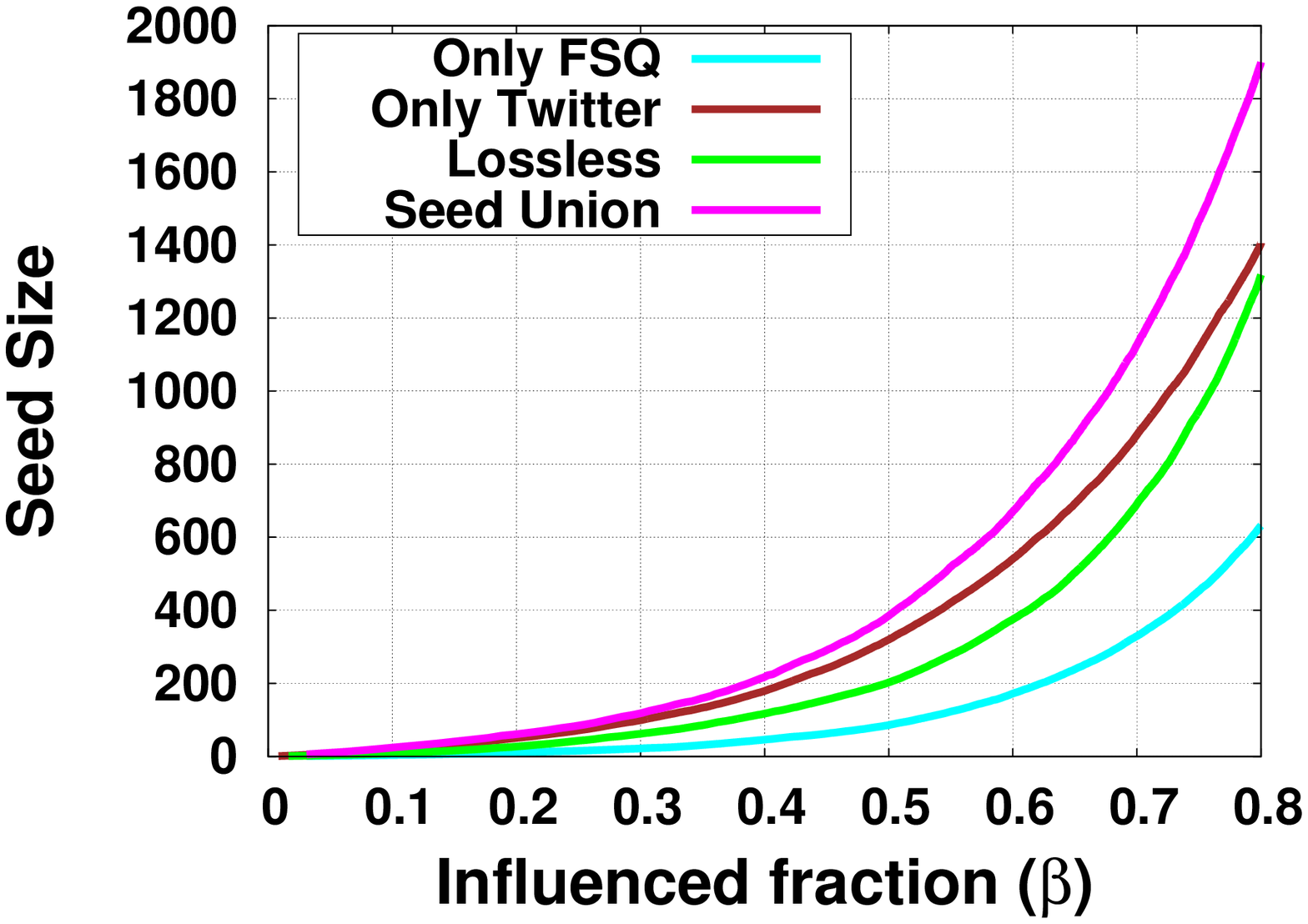}
	}\\
	\subfigure[Co-author networks] {
	\includegraphics[width=0.4\columnwidth]{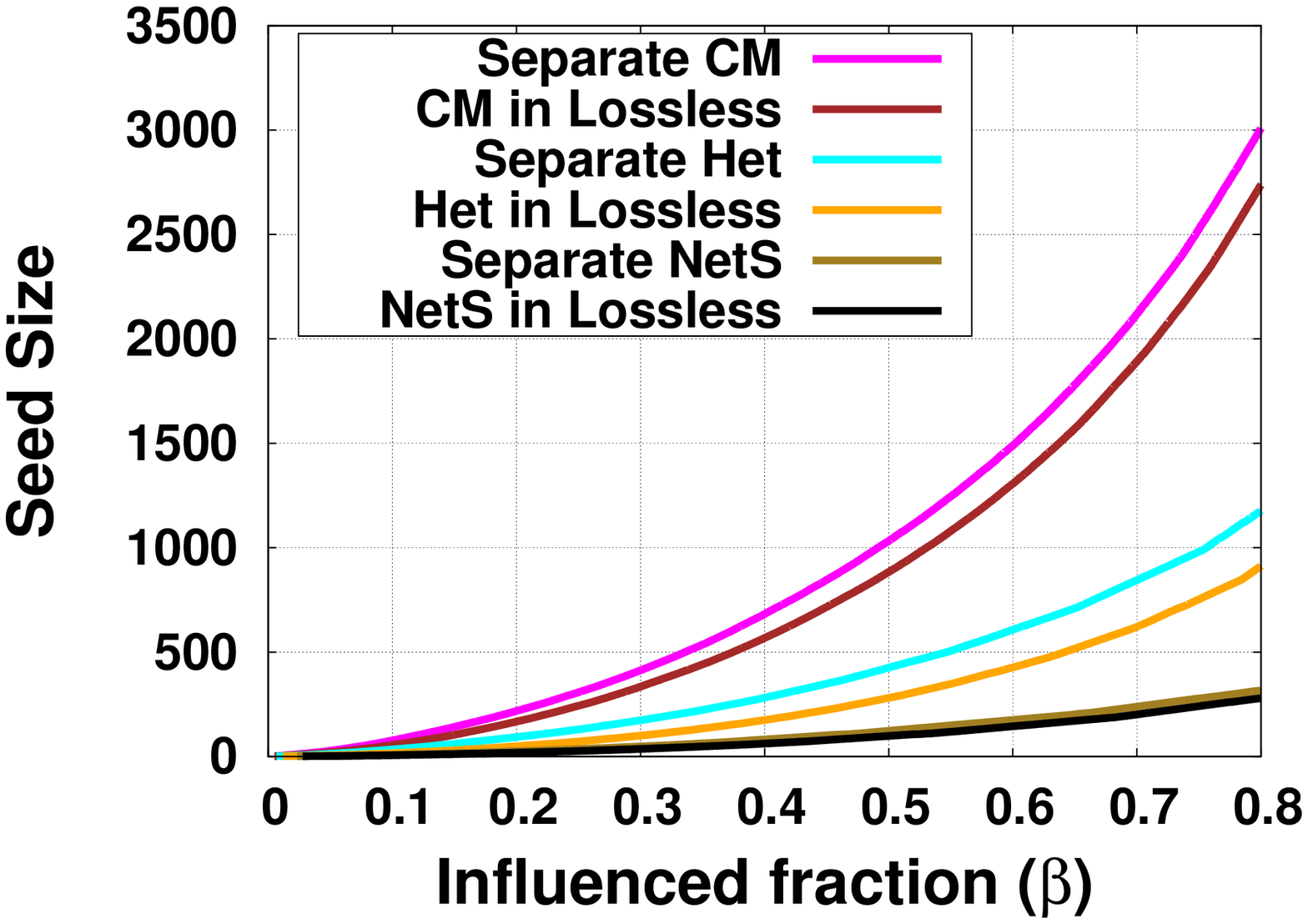}
	}
	\subfigure[FSQ and Twitter] {
	\includegraphics[width=0.4\columnwidth]{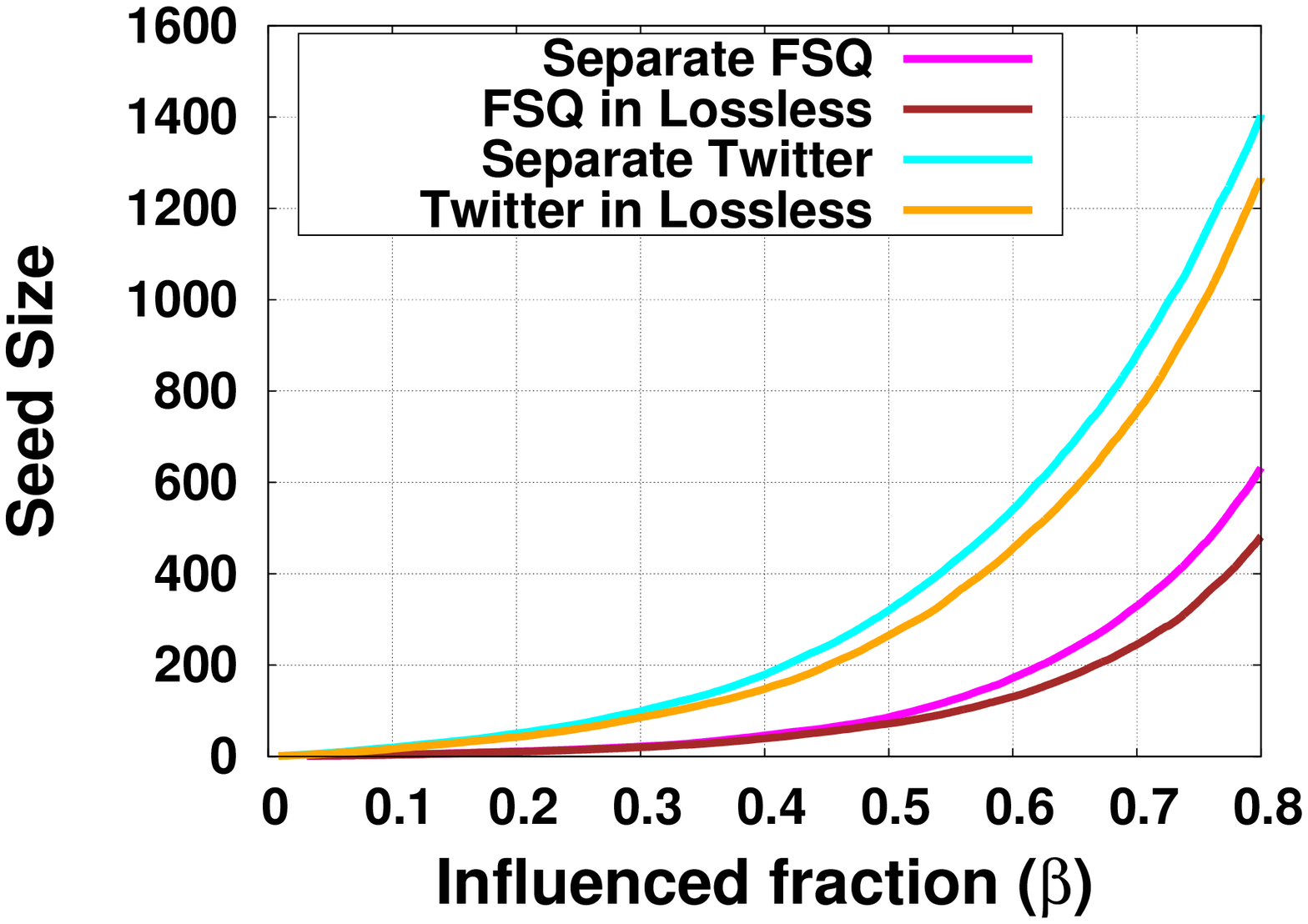}
	}
\caption{The quality of seed sets with and without using the coupled network}
\label{fig:Individual_VS_Coupled}
\end{figure}

In the second scenario, the lossless coupling scheme achieves the best result in both networks.  When the network is considered as a standalone network and choose seeds individually(labeled with Only in Fig. \ref{fig:Individual_VS_Coupled}), the seeds size is relatively larger than choosing from the coupled network. As shown in Fig. \ref{fig:Individual_VS_Coupled}), the sizes decrease by 9\%, 25\%, and 17\%  in CM, Het, and FSQ, accordingly. This improvement is also due to the information diffusion across several networks by the overlapping users. Especially, when the network sizes are unbalanced, like Het with the smaller size of users seems to get more  improvement than CM.

\subsection{Analysis of seed sets}
We analyze seed sets with different influenced fraction $\beta$ to find out: the composition of the seed set and the influenced set; and the influence contribution of each network.
As illustrated in Fig. \ref{fig:coupling_bias}, a significant fraction of the seed set is overlapping nodes although only 5\% (7\%) users of FSQ-Twitter (the co-author networks) are overlapping users. 
With $\beta = 0.4$, the fraction of overlapping seed vertices is around 24.9\% and 25\% in the co-author and FSQ-Twitter networks, respectively. 
As overlapping users can influence friends in different networks, they are more likely to be selected in the seed set than ones participating in only one network. 
Fig. \ref{fig:relay_constribution} demonstrates the high influence contribution of the overlapping users, especially when $\beta$ is small (contribute more than 50\% of the total influence when $\beta = 0.2$).
However, when $\beta$ is large, good overlapping users are already selected, so overlapping users are not favored any more. 

\begin{figure}[h]%
	\centering
	\subfigure{
	\includegraphics[width=0.4\columnwidth]{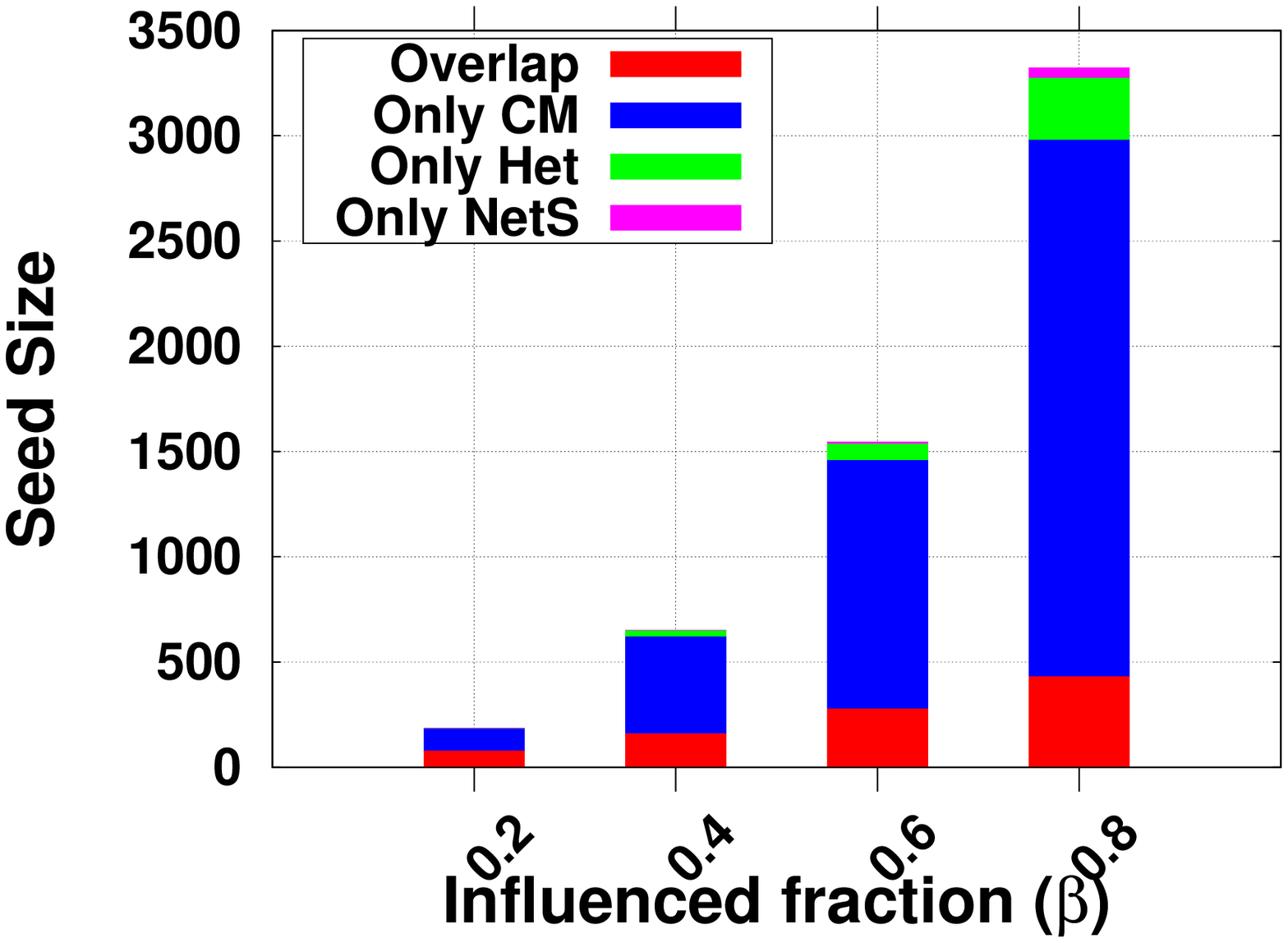}
	}	
	\subfigure {
	\includegraphics[width=0.4\columnwidth]{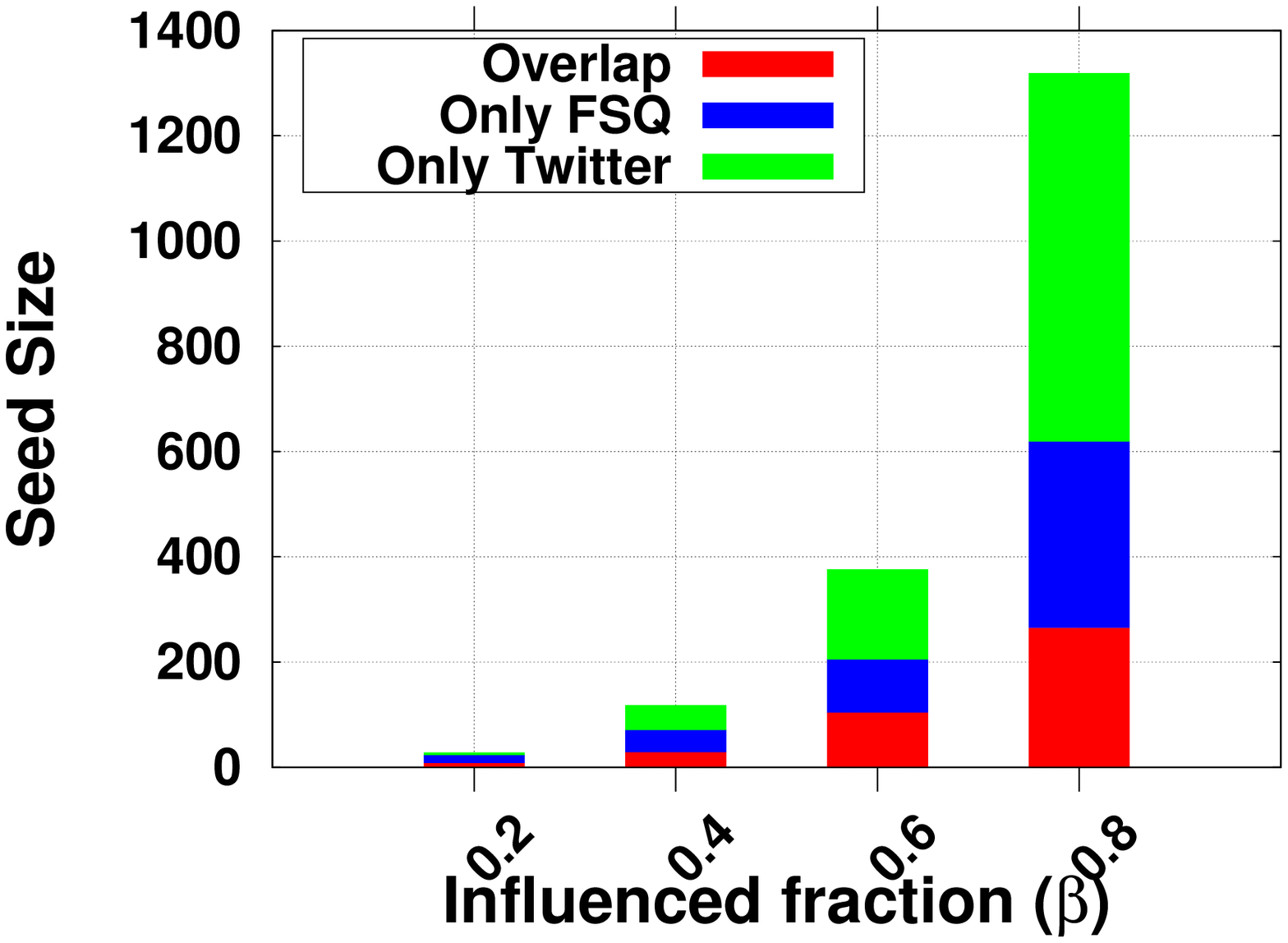}
	}
	\setcounter{subfigure}{0}
	\subfigure[Co-author networks]{
	\includegraphics[width=0.4\columnwidth]{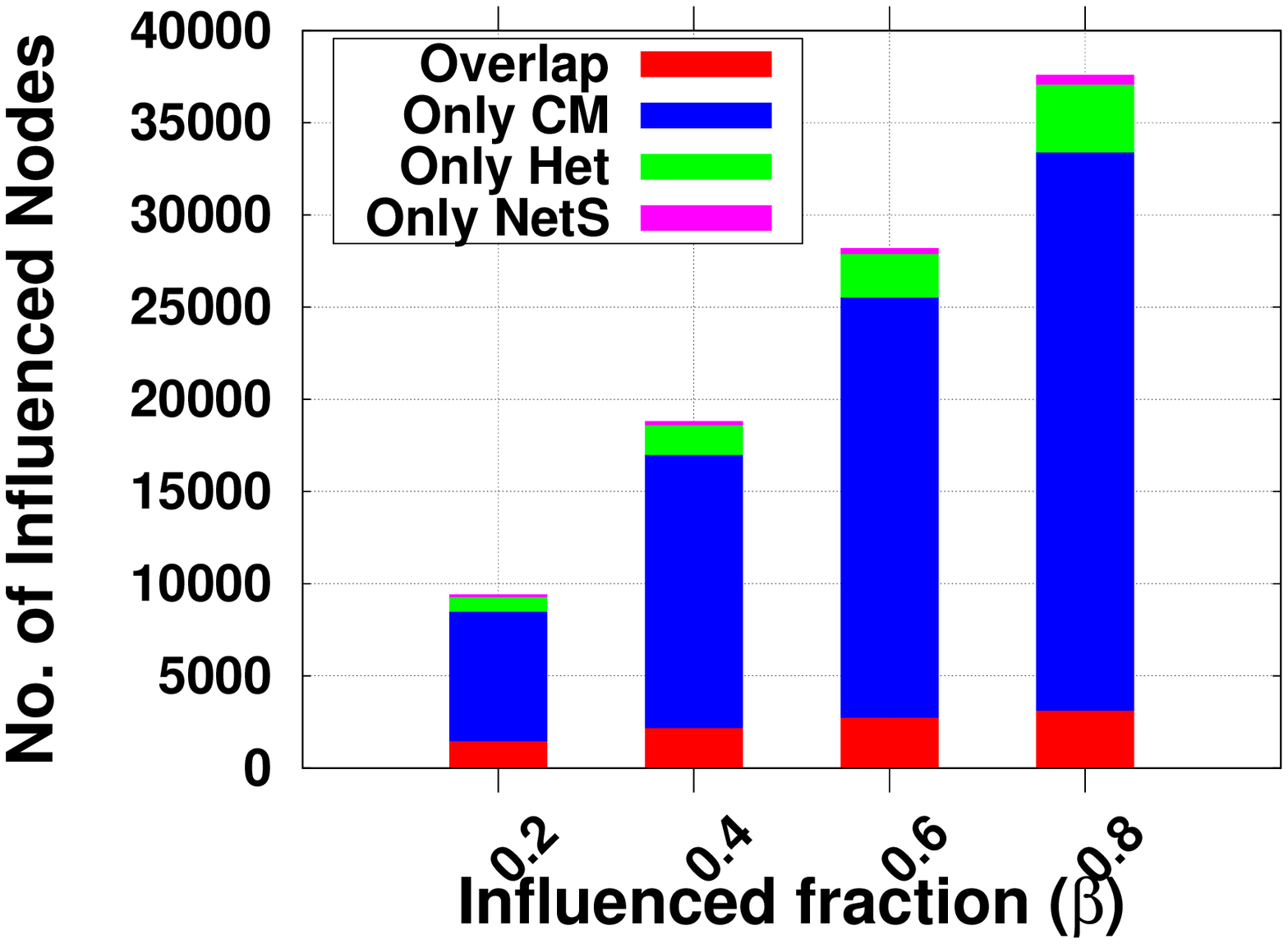}
	}
	\subfigure[FSQ-Twitter]{
	\includegraphics[width=0.4\columnwidth]{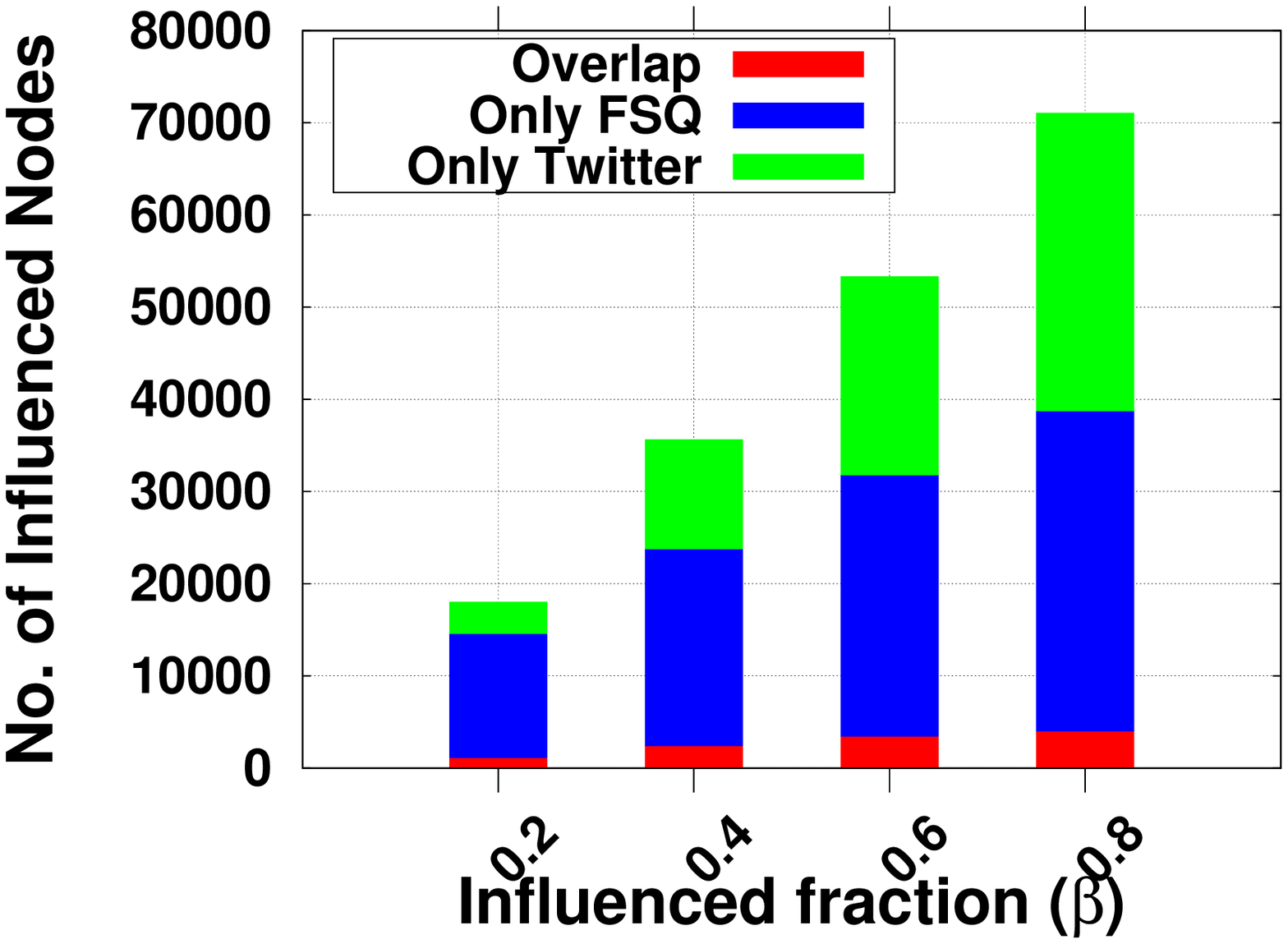}
	}
\caption{The bias in selecting seed nodes}
\label{fig:coupling_bias}
\end{figure}

Additionally, there is an imbalance between the number of selected vertices and influenced vertices in each networks.
In the co-author dataset, CM contributes a large number of seed vertices and influenced vertices since the size of CM is significantly larger than other networks. 
When $\beta = 0.8$, 76.7\% of seed vertices and 80.5\% of influenced vertices are from CM.
In contrast, the number of seed vertices from FSQ is small but the number of influenced vertices in FSQ is much higher than Twitter. 
With $\beta = 0.4$, 27\% (without overlapping vertices) of seed vertices belong to FSQ while 70\% of influenced vertices are in FSQ. 
After the major of vertices in FSQ are influenced, the algorithm starts to select more vertices in Twitter to increase the influence fraction.
This implies that it is easier for the information to propagate in one network than the other, even when we consider the overlapping between them.
Moreover, we can target the overlapping users in one network (e.g. Twitter) to influence users in another network (e.g. FSQ).

\subsection{Mutual impact of networks}
We evaluate the mutual impact between networks when the number of network $k$ increases.
We use a user base of 10000 users to synthesize networks for the experiment. For each network, we randomly select 4000 users from the user base and connect each pair of selected users randomly with probability 0.0025. Thus all networks have the same size and the expected average outgoing (incoming) degree of 10. The expected overlapping fraction of any network pair is 16\%. We measure the seed size to influence 60\% of users (6000 users) with the different number of networks (Fig. \ref{fig:impact_seedsize}). When $k$ increases from 2 to 5, the seed size decreases several times. It implies that the introduction of a new OSN increases the diffusion of information significantly.
\vspace{10pt}

\begin{figure}[h]%
  \centering
	\subfigure[Co-author networks] {
	\includegraphics[width=0.4\columnwidth]{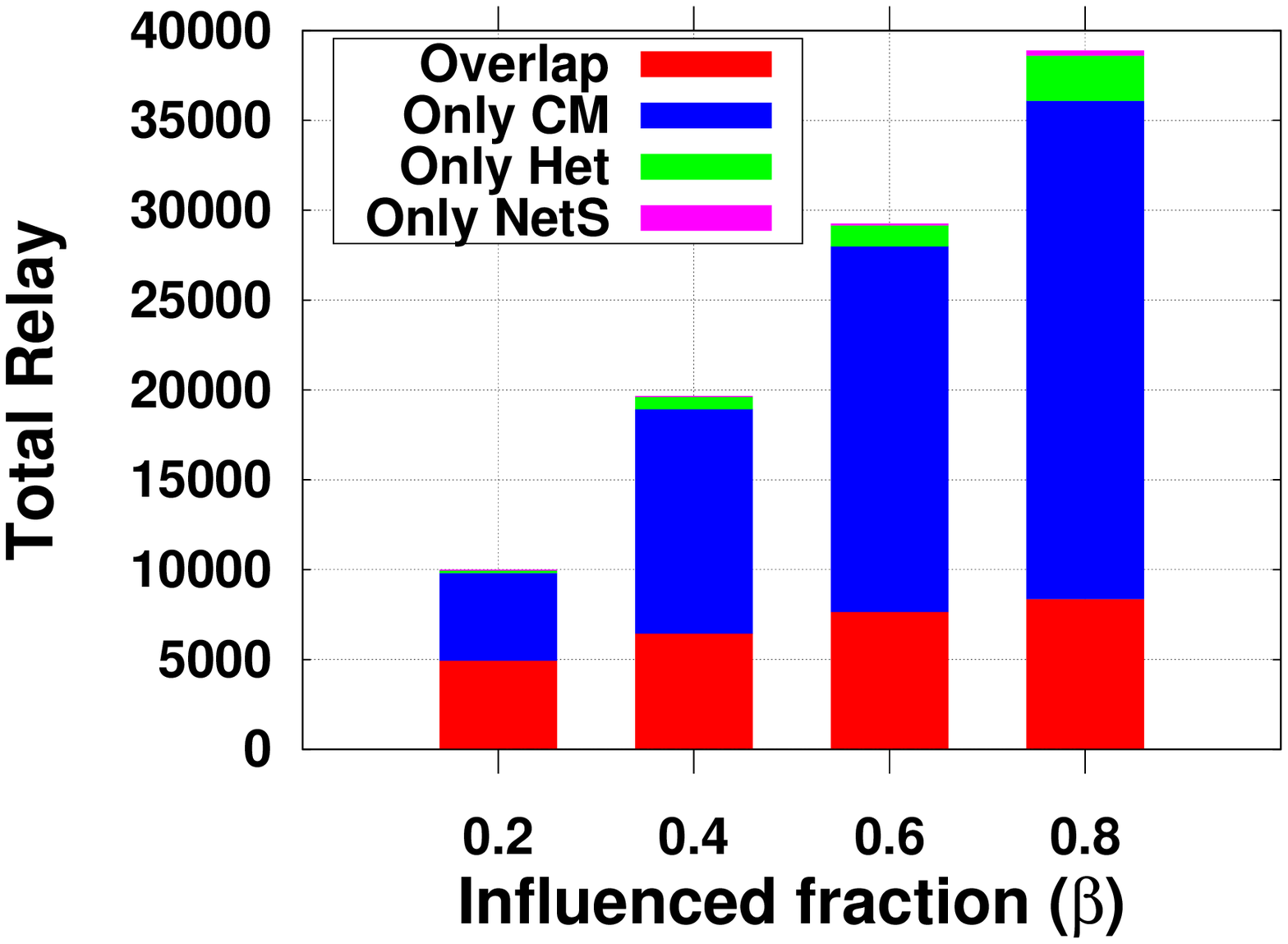}
	}
	\subfigure [FSQ and Twitter]{
	\includegraphics[width=0.4\columnwidth]{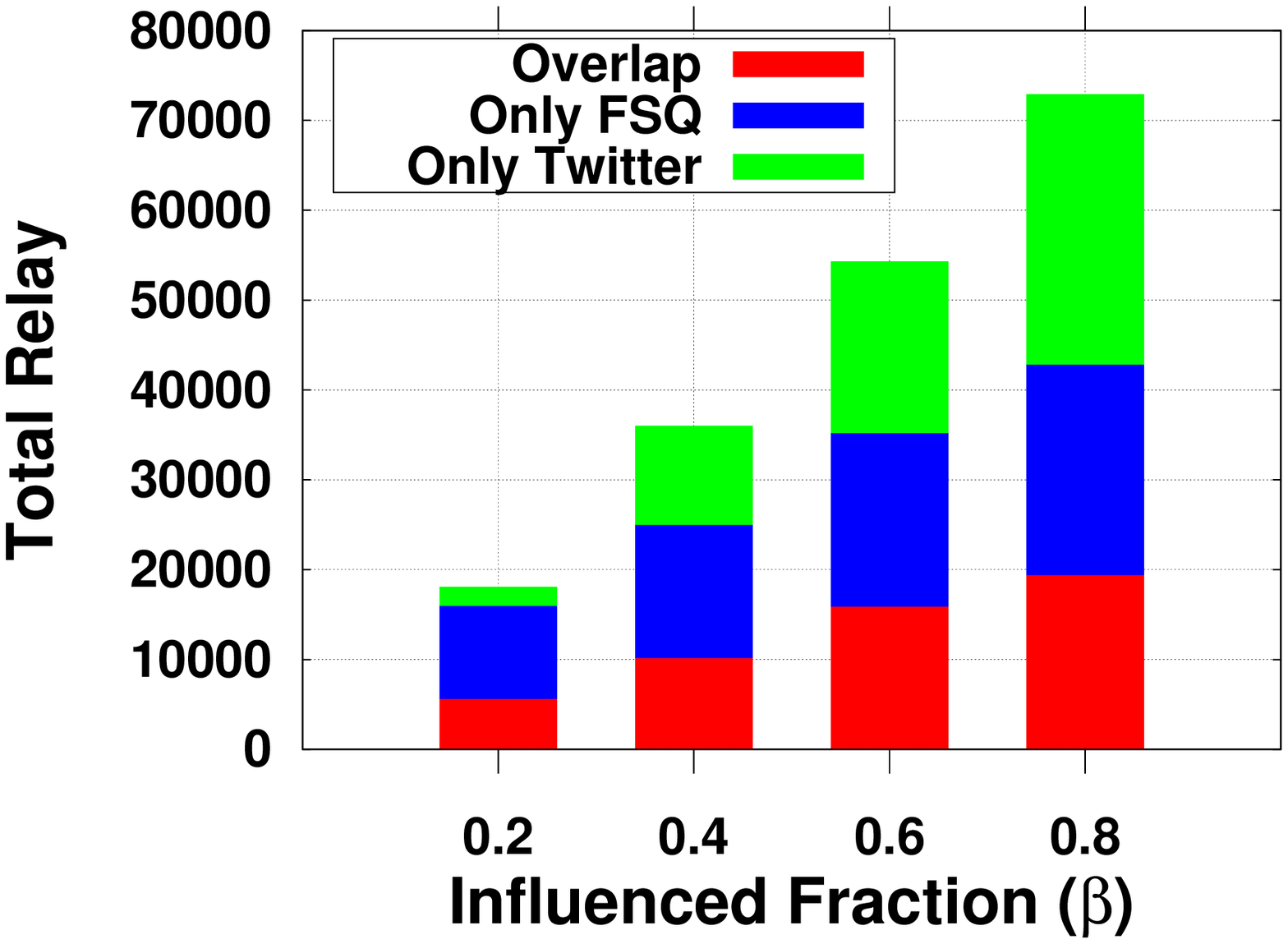}
	}	
\caption{The influence contribution of seed vertices from component networks}
\label{fig:relay_constribution}
\end{figure}

We also compute how much new networks help the existing one to propagate the information. 
Using the same seed set found by the greedy algorithm to influence 60\% (2400 users) of the target network (the first created network), we compute the total number of influenced vertices in that network as well as the external influence. Fig. \ref{fig:impact_external} shows that the number of influenced vertices is raised 46\% with the support of 3 new networks when $k$ is changed from 2 to 5. In addition, the fraction of external influence is also increased dramatically from 39\% when $k = 2$ to 67\% when $k = 5$. It means that the majority of influence can be obtained via the support of other networks. On the hand, these results suggest that the existing networks may benefit from the newly introduced competitor. 

\begin{figure}[h]
  \centering
	\subfigure[] {
  \includegraphics[width=0.4\columnwidth]{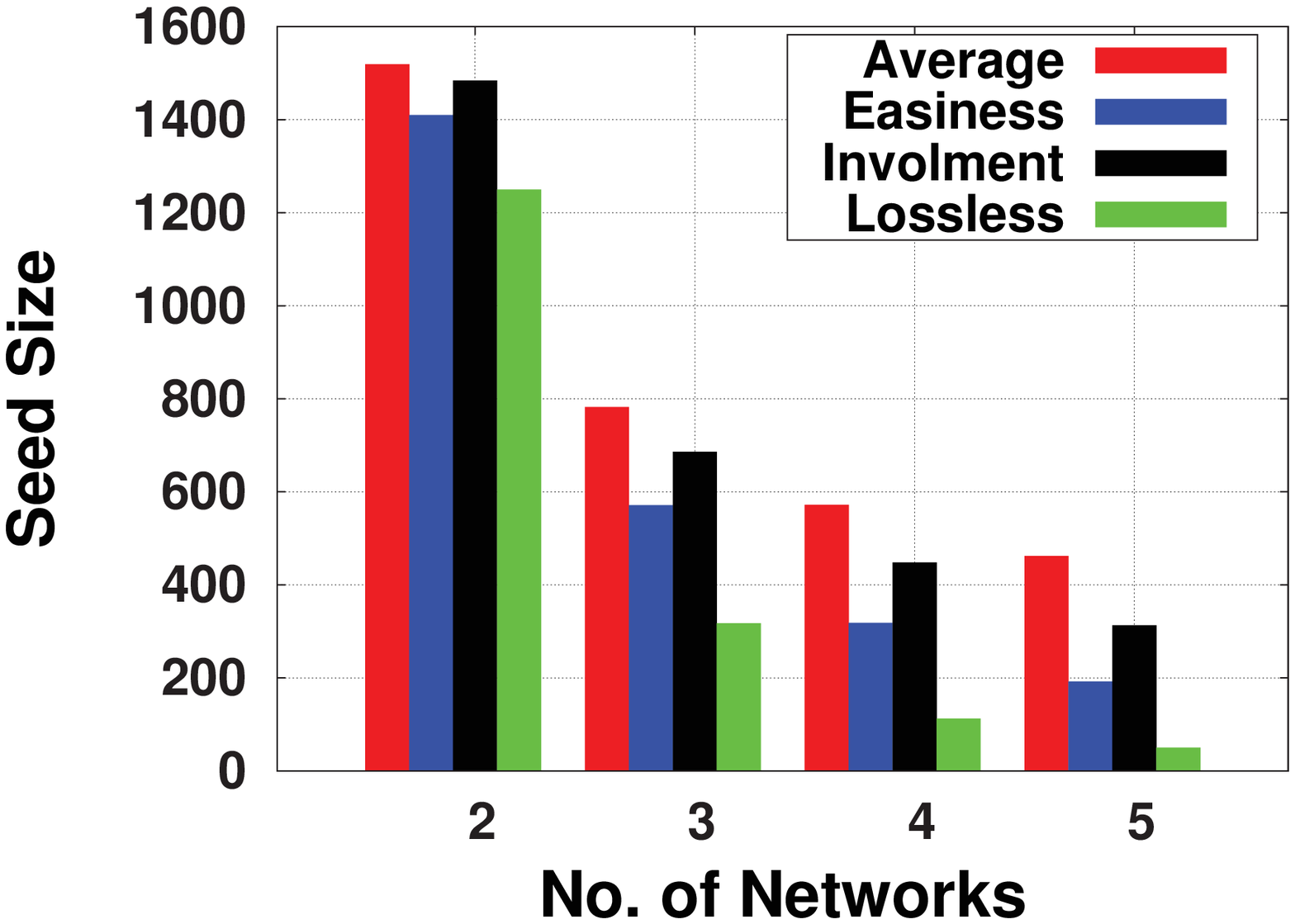}
  \label{fig:impact_seedsize}
	}
	\subfigure[]{
	\includegraphics[width=0.4\columnwidth]{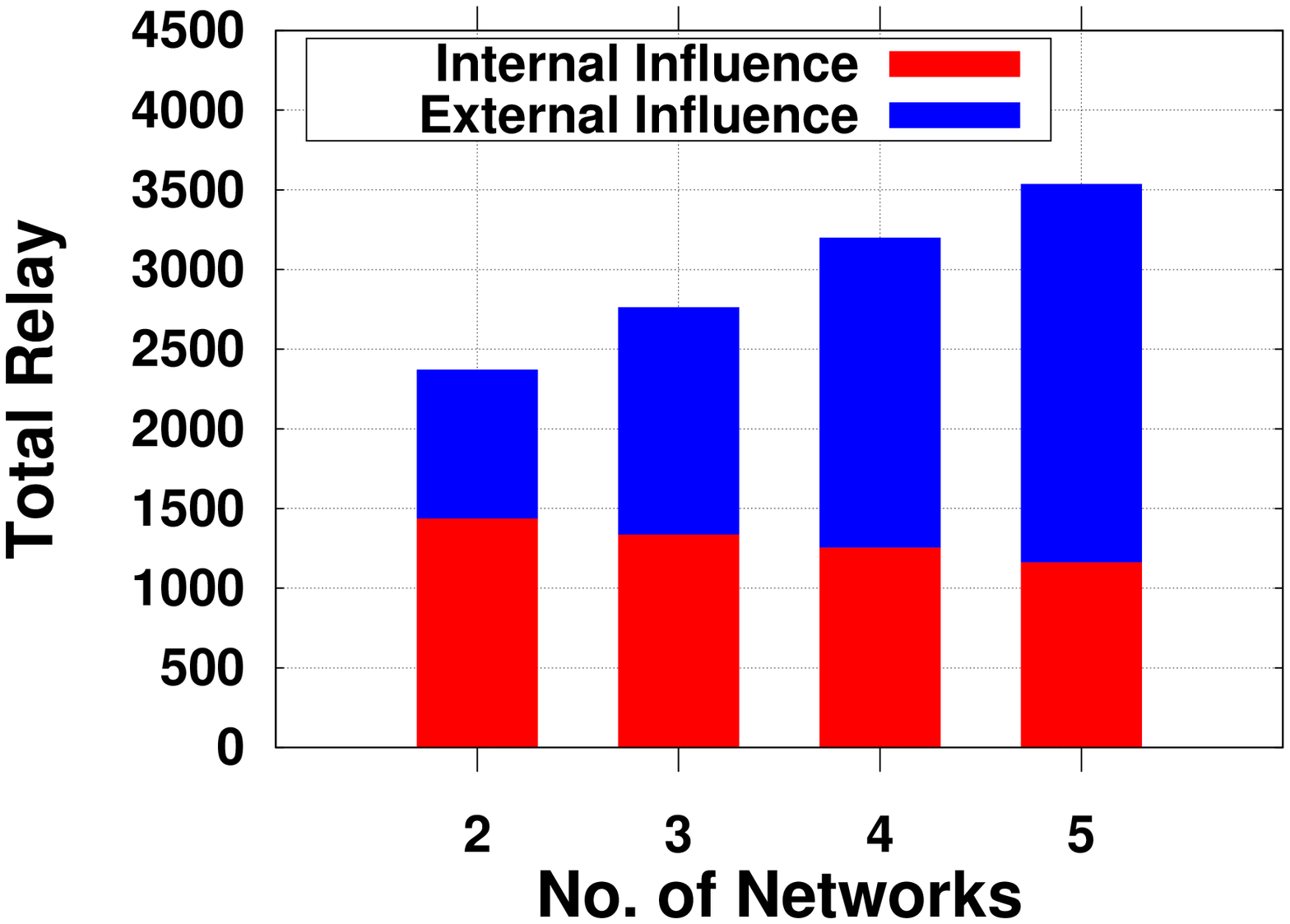}
	\label{fig:impact_external}
	}
\caption{The impact of additional networks }
\label{fig:new_networks}
\end{figure}

\section{Conclusions}
\label{se:conclusion}
In this paper, we study the least cost influence problem in multiplex networks. To tackle the problem, we introduced novel coupling schemes to reduce the problem to a version on a single network. Then we design a new metric to quantify the flow of influence inside and between networks based on the coupled network. Exhaustive experiments provide new insights to the information diffusion in multiplex networks.

In the future, we plan to investigate the problem in multiplex networks with heterogeneous diffusion models in which each network may have its own diffusion model. It is still an ongoing problem whether they can be represented efficiently, or we have better method to couple them into one network.

\section*{Acknowledgment}
This work is supported in part of NSF CCF-1422116 and DTRA HDTRA1-14-1-0055.

\bibliographystyle{plain}
\bibliography{ref.bbl}

\end{document}